\documentclass[10pt, conference, letterpaper]{IEEEtran}

\usepackage{cite}
\usepackage{color}
\usepackage{graphicx}

\usepackage{amssymb}
\usepackage{amsfonts}
\usepackage[vlined,linesnumbered,ruled,norelsize]{algorithm2e}
\usepackage[usenames,dvipsnames]{xcolor}
\usepackage{bm}
\usepackage{array}
\usepackage{subfigure}
\usepackage{float}
\usepackage[switch]{lineno}
\usepackage{url}
\usepackage{xcolor}
\usepackage{wrapfig}
\usepackage{mdwlist}
\usepackage{stmaryrd}
\usepackage{enumerate}
\usepackage{amsmath,amsthm}
\usepackage{pgfplots}
\usepackage{makecell}
\usepackage{caption2}

\newtheorem{theorem}{Theorem}
\newtheorem{lemma}{Lemma}
\newtheorem{definition}{Definition}
\newtheorem{fact}{\textbf{Fact}}

\begin{document}
%
\title{Cost-Effective Seed Selection for Online Social Networks}

\author{\IEEEauthorblockN{Kai Han\IEEEauthorrefmark{1}~~~~Yuntian He\IEEEauthorrefmark{1}~~~~Xiaokui Xiao\IEEEauthorrefmark{2}~~~~Shaojie Tang\IEEEauthorrefmark{3}~~~~Jingxin Xu\IEEEauthorrefmark{1}~~~~Liusheng Huang\IEEEauthorrefmark{1}}
\IEEEauthorblockA{\IEEEauthorrefmark{1}School of Computer Science and Technology, University of Science and Technology of China, China\\
\IEEEauthorrefmark{2}School of Computer Science and Engineering, Nanyang Technological University, Singapore\\
\IEEEauthorrefmark{3}Naveen Jindal School of Management, University of Texas at Dallas, United States\\
}}

\maketitle

\begin{abstract}
We study the min-cost seed selection problem in online social networks, where the goal is to select a set of seed nodes with the minimum total cost such that the expected number
of influenced nodes in the network exceeds a predefined threshold. We propose several algorithms that outperform the state-of-the-art algorithms both theoretically and experimentally. Under the case where the users have heterogeneous costs, our algorithms are the first bi-criteria approximation algorithms with polynomial running time and provable logarithmic performance bounds using a general contagion model. Under the case where the users have uniform costs, our algorithms achieve logarithmic approximation ratio and provable time complexity which is smaller than that of existing algorithms in orders of magnitude. We conduct extensive experiments using real social networks. The experimental results show that, our algorithms significantly outperform the existing algorithms both on the total cost and on the running time, and also scale well to billion-scale networks.
\end{abstract}


\section{Introduction}
\label{sec:intro}

With the proliferations of Online Social Networks (OSNs) such as Facebook and LinkedIn, the paradigm of viral marketing through the ``word-of-mouth'' effect over OSNs has found numerous applications in modern commercial activities. For example, a company may provide free product samples to a few individuals (i.e, ``seed'' nodes) in an OSN, such that more people can be attracted to buy the company's products through the information cascade starting from the seed nodes.

Kempe~\textit{et.al.}~\cite{Kempe2003} have initiated the studies on the NP-hard $k$-Seed-Selection ($k$SS) problem in OSNs, where the goal is to select $k$ most influential nodes in an OSN under some contagion models such as the independent cascade model and the more general triggering model. After that, extensive studies in this line have appeared to design efficient approximation algorithms for the $k$SS problem and its variations~\cite{LeskovecKGFVG2007,ChenWW2010,Borgs2014,TangXS2014,TangSX2015,NguyenTD2016,CohenDPW2014,Nguyen2017outward,NguyenDT2016,LinCL17}.
However, all the algorithms proposed in these studies can be classified as Influence Maximization (IM) algorithms, because they have the same goal of optimizing the \textit{influence spread} (i.e., the expected number of influenced nodes in the network).


In many applications for viral marketing, one may want to seek a ``best bang for the buck'', i.e., to select a set $S$ of seed nodes with the minimum total cost such that the influence spread (denoted by $f(S)$) is no less than a predefined threshold $\eta$. This problem is called as the Min-Cost Seed Selection (MCSS) problem and has been investigated by some prior work such as~\cite{Chen2014,Kuhnle2017,Goyal2013}. It is indicated by these work that the existing IM algorithms are not appropriate for the MCSS problem, as the IM algorithms require the knowledge on the total cost of the selected seed nodes in advance, while this knowledge is exactly what we pursue in the MCSS problem.

The existing MCSS algorithms~\cite{Chen2014,Kuhnle2017,Goyal2013}, however, suffer from several major deficiencies. First, these algorithms only consider the \textbf{uniform cost (UC)} case where each node has an identical cost 1, so none of their performance ratios holds under the \textbf{general cost (GC)} case where the nodes' costs can be heterogeneous. Nevertheless, the GC case has been indicated by existing work to be ubiquitous in reality~\cite{NguyenDT2016,Goyal2013,NguyenZ2013}. Second, most of them only propose bi-criteria approximation algorithms, which cannot guarantee that the influence spread of the selected nodes is no less than $\eta$. Third, the theoretical bounds proposed by some existing work only hold for OSNs with special influence propagation probabilities, but do not hold for general OSNs. Last but not the least, the existing MCSS algorithms do not scale well to big networks with millions/billions of edges/nodes.

Although the MCSS problem looks like a classical Submodular Set Cover (SSC) problem~\cite{Wolsey1982}, the conventional approximation algorithms for SSC cannot be directly used to find  solutions to MCSS in polynomial time, as computing the influence spread of any seed set is essentially a \#P-hard problem~\cite{ChenWW2010}. One possible way for overcoming this hurdle is to apply the existing network sampling methods proposed for the $k$SS problem (e.g., \cite{Borgs2014,TangSX2015,NguyenDT2016}), but it is highly non-trivial to design an efficient sampling algorithm to get a satisfying approximation solution to MCSS, due to the essential discrepancies between MCSS and $k$SS. Indeed, the number of selected nodes in the $k$SS problem is always pre-determined (i.e., $k$), while this number can be uncertain and highly correlated with the generated network samples in MCSS. This requires us to carefully build a quantitative relationship between the generated network samples and the stopping rule for selecting seed nodes, such that a feasible solution with small cost can be found in as short time as possible. Unfortunately, the existing studies have not made a substantial progress towards tackling the above challenges in MCSS, so they suffer from several deficiencies on the theoretical performance ratio, practicability and efficiency, including the ones described in last section.

%

In this paper, we propose several algorithms for the MCSS problem based on network sampling, and our algorithms advance the prior-art by achieving better approximation ratios and lower time complexity. More specifically, our major contributions include: 
\begin{enumerate}
\item {Under the GC case, we propose the first polynomial-time bi-criteria approximation algorithms with provable approximation ratio for MCSS, using a general contagion model (i.e., the triggering model~\cite{Kempe2003}). Given any $\alpha,~\delta\in (0,1)$ and any OSN with $n$ nodes and $m$ edges, our algorithms achieve an $\mathcal{O}(\ln\frac{1}{\alpha})$ approximation ratio and output a seed set $S$ satisfying $f(S)\geq (1-\alpha)\eta$ with the probability of at least $1-\delta$. Our algorithms also achieve an expected time complexity of $\mathcal{O}(\frac{m q}{\alpha^2}\ln\frac{n}{(1-\alpha)\delta\eta})$, where $q$ is the maximum influence spread of any single node in the network.}
\item Under the UC case, our proposed algorithms have $\mathcal{O}(\ln\frac{n\eta}{n-\eta})$ approximation ratios and can output a seed set $S$ satisfying $f(S)\geq \eta$ with probability of at least $1-\delta$. Compared to the existing algorithm~\cite{Chen2014} that has the best-known approximation ratio under the same setting with ours, the running time of our algorithms is at least $\Omega(\frac{n^2}{\ln n})$ times faster, while our approximation ratio can be better than theirs under the same running time.
\item In contrast to some related work such as~\cite{Kuhnle2017}, the performance bounds of all our algorithms do not depend on particularities of the network data (e.g., the influence propagation probabilities of the network), so they are more general.
\item We test the empirical performance of our algorithms using real OSNs with up to 1.5 billion edges. The experimental results demonstrate that our algorithms significantly outperform the existing algorithms both on the running time and on the total cost of the selected seed nodes.
\end{enumerate}

The rest of our paper is organized as follows. We formally formulate our problem in Sec.~\ref{sec:prelim}. We propose bi-criteria approximation algorithms and approximation algorithms for the MCSS problem in Sec.~\ref{sec:bicriteria} and Sec.~\ref{sec:approxforuc}, respectively. The experimental results are presented in Sec.~\ref{sec:pe}. We review the related work in Sec.~\ref{sec:rw} before concluding our paper in Sec.~\ref{sec:conclu}. 

%
%
%
%
%



\section{Preliminaries}
\label{sec:prelim}

\subsection{Model and Problem Definition}

We model an online social network as a directed graph $G=(V,E)$ where $V$ is the set of nodes and $E$ is the set of edges. Each node $u\in V$ has a cost $C(\{u\})$ which denotes the cost for selecting $u$ as a seed node. For convenience, we define $C(A)=\sum_{u\in A}C(\{u\})$ for any $A\subseteq V$.

When the nodes in a seed set $S\subseteq V$ are influenced, an influence propagation is caused in the network and hence more nodes can be activated. There are many influence propagation/contagion models, among which the Independent Cascade (IC) model and the Linear Threshold (LT) model~\cite{Kempe2003} are the most popular ones. However, it has been proved that both the IC model and the LT model are special cases of the triggering model~\cite{Kempe2003}, so we adopt the triggering model for generality.

In the triggering model, each node $u\in V$ is associated with a \textit{triggering distribution} $\mathcal{I}(u)$ over $2^{N_{in}(u)}$ where $N_{in}(u)=\{v\in V: \langle v,u\rangle\in E\}$. Let $I(u)$ denote a sample taken from $\mathcal{I}(u)$ for any $u\in V$ ($I(u)$ is called a \textit{triggering set} of $u$). The influence propagation with any seed set $A$ under the triggering model can be described as follows. At time $0$, the nodes in $A$ are all activated. Afterwards, any node $u$ will be activated at time $t+1$ iff there exists a node $v\in I(u)$ which has been activated at time $t$. This propagation process terminates when no more nodes can be activated. Let $f(A)(\forall A\subseteq V)$ denote the {Influence Spread} (IS) of $A$, i.e., the expected number of activated nodes under the triggering model. The problem studied in this paper can be formulated as follows:

\begin{definition}[The MCSS problem]
Given an OSN $G=(V,E)$ with $|V|=n$ and $|E|=m$, a cost function $C$, and any $\eta\in (0,n)$, the Min-Cost Seed Selection (MCSS) problem aims to find $S_{opt}\subseteq V$ such that $f(S_{opt})\geq \eta$ and $C(S_{opt})$ is minimized.
\end{definition}

The MCSS problem has been studied by prior work under different settings/assumptions such as the UC setting and the GC setting explained in Sec.~\ref{sec:intro}. Besides, under the \textbf{Exact Value (EV)} setting, it is assumed that the exact value of $f(A)(\forall A\subseteq V)$ can be computed in polynomial time, while this assumption does not hold under the \textbf{Noisy Value (NV)} setting.

The existing MCSS algorithms can also be classified into: (1) \textbf{APproximation (AP) algorithms}~\cite{Chen2014}: these algorithms regard any $A\subseteq V$ satisfying $f(A)\geq \eta$ as a feasible solution; (2) \textbf{Bi-criteria Approximation (BA) algorithms}~\cite{Kuhnle2017,Goyal2013}: these algorithms regard any $A\subseteq V$ satisfying $f(S)\geq (1-\alpha)\eta$ as a feasible solution, where $\alpha$ is any given number in $(0,1)$.

\subsection{The Greedy Algorithm for Submodular Set Cover}
\label{sec:gressc}

It has been shown that $f(\cdot)$ is a \textit{monotone and submodular} function~\cite{Kempe2003}, i.e., for any $S_1\subseteq S_2\subseteq V$ we have $f(S_1)\leq f(S_2)$ and $\forall x\in V\backslash S_2: f(S_1\cup \{x\})-f(S_1)\geq f(S_2\cup \{x\})-f(S_2)$. Therefore, the MCSS problem is an instance of the Submodular Set Cover (SSC) problem~\cite{Wolsey1982,Wan2010}, which can be solved by a greedy algorithm under the EV setting. For clarity, we present a (generalized) version of the greedy algorithm for the SSC problem, shown in Algorithm~\ref{alg:naivegreedy}:

 \begin{algorithm}[tp]
  $A\leftarrow \emptyset;~\hat{f}(A)\leftarrow 0$\\
  \While{$\hat{f}(A) < \Phi $}
   {
        \ForEach{$u\in V\backslash A$}{
            $A'\leftarrow A\cup \{u\}$\\
            Compute $\hat{f}(A')$ such that $\mathbb{P}\{(1-\gamma_1){f}(A')\leq \hat{f}(A')\leq (1+\gamma_2){f}(A')\}\geq 1-\theta$ \label{ln:estimateinfluence}
        }
          $u^*\leftarrow \arg\max_{u\in V\backslash A}\frac{\min\{\hat{f}(A\cup \{u\}),\Phi\}-\hat{f}(A)}{C(\{u\})}$;\\
          $A\leftarrow A\cup \{u^*\}$
   }
  \Return{$A$} \label{ln:returnfinalresult}
    \caption{$\mathsf{GreSSC}(\Phi,\gamma_1,\gamma_2,\theta)$} \label{alg:naivegreedy}
\end{algorithm}


Under the EV setting (i.e.,  $\gamma_1=\gamma_2=\theta=0$), $\mathsf{GreSSC}$ runs in polynomial time and also achieves nice performance bounds for MCSS~\cite{ZhuLZ2016,Goyal2013}. For example, we have:


\begin{fact}[\cite{Goyal2013}]
Let $S'$ be the output of $\mathsf{GreSSC}((1-\alpha)\eta,0,0,0)$ for any $\alpha\in (0,1)$. Then we have $f(S')\geq (1-\alpha)\eta$ and $C(S')\leq (1+\ln\frac{1}{\alpha})C(S_{opt})$. This bound holds even if $f(\cdot)$ is an arbitrary monotone and non-negative submodular function defined on the ground set $V$.
\label{fct:goyalgreedy}
\end{fact}


Unfortunately, it has been shown that calculating $f(A)$ $(\forall A\subseteq V)$ is \#P-hard~\cite{ChenWW2010}. Therefore, implementing $\mathsf{GreSSC}$ under the EV setting is impractical. Indeed, the existing work usually uses multiple monte-carlo simulations to compute $\hat{f}(A')$ in Algorithm~\ref{alg:naivegreedy}~\cite{Kempe2003,LeskovecKGFVG2007}, so $\hat{f}(A')$ is a noisy estimation of $f(A')$. However, such an approach has two drawbacks: (1) The time complexity of $\mathsf{GreSSC}$ is still high (though polynomial); (2) The theoretical performance bounds of $\mathsf{GreSSC}$ under the EV setting (such as Fact~\ref{fct:goyalgreedy}) no longer hold.

\subsection{RR-Set Sampling}

Recently, Borgs. \textit{et. al.}~\cite{Borgs2014} have proposed an elegant network sampling method to estimate the value of $f(A) (\forall A\subseteq V)$, whose key idea can be presented by the following equation:
\begin{eqnarray}
\forall A\subseteq V: f(A)=n\cdot \mathbb{P}\{R\cap A\neq\emptyset\},
\label{eqn:borgsequation}
\end{eqnarray}
where $R$ is a random subset of $V$, called as a Reverse Reachable (RR) set. Under the IC model~\cite{Borgs2014}, an RR-set can be generated by first uniformly sampling $u$ from $V$, then reverse the edges' directions in $G$ and traverse $G$ from $u$ according to the probabilities associated with each edge. According to equation~(\ref{eqn:borgsequation}), the value of $f(A)(\forall A\subseteq V)$ can be estimated unbiasedly by any set $\mathcal{U}$ of RR-sets as follows:
\begin{eqnarray}
\bar{f}(\mathcal{U},A)={n\cdot\sum\nolimits_{R\in \mathcal{U}} X(R, A) }/{|\mathcal{U}|}
\label{eqn:estimation}
\end{eqnarray}
where $X(R,A)\triangleq \min\{1,|R\cap A|\}$.
It is noted that the RR-set sampling method can also be applied to the Triggering model, where equations~(\ref{eqn:borgsequation})-(\ref{eqn:estimation}) still hold~\cite{TangXS2014}.

The work of Borgs. \textit{et. al.}~\cite{Borgs2014} and other proposals~\cite{TangXS2014,TangSX2015,NguyenTD2016} have shown that the influence maximization problem can be efficiently solved by the RR-set sampling method. Nevertheless, how to use this method to solve the MCSS problem still remains largely open.


    \begin{algorithm} [tp!]
    \KwIn{$ G=(V,E), \delta, \alpha,\sigma,\gamma, \mu, \eta$ }
    \KwOut{A set $S\subseteq V$ satisfying $f(S)\geq (1-\alpha)\eta$ w.h.p.}
    %
  $W_1\leftarrow \langle (1-\alpha)\eta,\frac{\gamma}{1-\alpha},\frac{\delta}{2\mu} \rangle$; $\Lambda\leftarrow (1-{\alpha}+\gamma)\eta$ \label{ln:setw1}\\
  $W_2\leftarrow \langle \eta,\sigma,\frac{\delta}{2} \rangle$;~$T\leftarrow \mathsf{SetT}(W_1,W_2)$ \label{ln:setw2}\\
  Generate a set $\mathcal{R}$ of RR-sets satisfying $|\mathcal{R}|=T$ \label{ln:generaterrsets}\\
  $S\leftarrow \mathsf{MCA}(\mathcal{R},\Lambda) $;~\Return{$S$} \label{ln:setbiglambda}
  \hrule
  \textbf{Function} $\mathsf{MCA}(\mathcal{R},\Lambda)$\\
 {$A\leftarrow \emptyset$} \\
 \lIf{$|V|<\Lambda$}{\Return $\emptyset$}
  \While{$\bar{f}(\mathcal{R},A)<\Lambda$}{
          $u^*\leftarrow \arg\max_{u\in V\backslash A}\frac{\min\{\bar{f}(\mathcal{R},A\cup \{u\}),\Lambda\}-\bar{f}(\mathcal{R},A)}{C(\{u\})}$;\\
          $A\leftarrow A\cup \{u^*\}$
  }
  \Return{$A$}
  \hrule \textbf{Function} $\mathsf{SetT}(W_1,W_2)$\\
  $T\leftarrow \lceil \max\{\mathrm{ut}(W_1), \mathrm{lt}(W_2)\} \rceil$;
%
  \Return{$T$}
    \caption{Bi-Criteria Approximation Algorithm for General Costs $(\mathsf{BCGC})$}
    \label{alg:bcgc}
  \end{algorithm}


\section{Bi-Criteria Approximation Algorithms}
\label{sec:bicriteria}

In this section, we propose bi-criteria approximation (BA) algorithms for the MCSS problem {under the GC+NV setting}.  

\subsection{The BCGC Algorithm}

Our first bi-criteria approximation algorithm is called $\mathsf{BCGC}$, shown in Algorithm~\ref{alg:bcgc}. $\mathsf{BCGC}$ first generates a set $\mathcal{R}$ of $T$ RR-sets (line~\ref{ln:generaterrsets}) according to the input variables $\delta, \alpha,\sigma,\gamma, \mu$ and $\eta$, and then calls the function $\mathsf{MCA}$ to find a min-cost node set $S$ satisfying $\bar{f}(\mathcal{R},S)\geq \Lambda$, which is returned as the solution to MCSS.

It can be verified that $\bar{f}(\mathcal{R},\cdot)$ is a monotone and non-negative submodular function defined on the ground set $V$, so $\mathsf{MCA}$ is essentially a (deterministic) greedy submodular set cover algorithm similar to $\mathsf{GreSSC}$. 


The key issue in $\mathsf{BCGC}$ is how to determine the values of $T$ and $\Lambda$. Intuitively, $T$ and $\Lambda$ should be large enough such that $\mathsf{MCA}$ can output a {feasible solution} $S$ (i.e., $f(S)\geq (1-\alpha)\eta$) with a  cost close to $C(S_{opt})$. On the other hand, we also want $T$ and $\Lambda$ to be small such that the time complexity of $\mathsf{BCGC}$ can be reduced. To see how $\mathsf{BCGC}$ achieve these goals, we first introduce the following functions: 
\begin{definition}
For any $\Gamma\in (0,n]$ and $\beta,\theta\in (0,1)$, define
\begin{eqnarray}
&&\mathrm{ut}(\langle \Gamma, \beta, \theta\rangle) \triangleq \min\left\{\frac{n^2}{2\beta^2\Gamma^2}\ln\frac{1}{\theta},\frac{2n(\beta+3)}{3\beta^2\Gamma}\ln\frac{1}{\theta}\right\};~~~~\label{eqn:defiofut}\\
&&\mathrm{lt}(\langle\Gamma, \beta, \theta\rangle) \triangleq [{2n }/({\beta^2 \Gamma})]\ln({1}/{\theta}); \label{eqn:defioflt}\\
&&\mathcal{Q}[\Gamma]\triangleq\{A\subseteq V: f(A)<\Gamma\};~D(\Gamma)\triangleq\left({{e}n}/{\lfloor \Gamma\rfloor}\right)^{\lfloor \Gamma\rfloor}, \label{eqn:defiofqgamma}
\end{eqnarray}
where $e$ is the base of natural logarithms.
\end{definition}
\noindent From Eqn.~(\ref{eqn:defiofqgamma}), it can be seen that
\begin{eqnarray}
|\mathcal{Q}[\Gamma]|\leq\sum\nolimits_{i=1}^{\lfloor \Gamma\rfloor} {n \choose i}\leq \sum\nolimits_{i=1}^{\lfloor \Gamma\rfloor}\frac{n^i}{i!}=\sum_{i=1}^{\lfloor \Gamma\rfloor}\frac{{\lfloor \Gamma\rfloor}^i}{i!}\frac{n^i}{\lfloor \Gamma\rfloor^i}\leq D(\Gamma)\nonumber
\end{eqnarray}
Moreover, (\ref{eqn:defiofut})-(\ref{eqn:defiofqgamma}) are useful for Lemmas~\ref{lma:utltbound}-\ref{lma:boundingoursolution}, which can be proved by the concentration bounds in probability theory~\cite{FanL2006}:




\begin{lemma}
Given any $A\subseteq V$ and any set $\mathcal{U}$ of RR-sets, if  $f(A)< \Gamma$ and $|\mathcal{U}|\geq \mathrm{ut}(\Gamma,\beta,\theta)$, then we have $\mathbb{P}\{\bar{f}(\mathcal{U},A)\geq (1+\beta)\Gamma\}\leq \theta$; if $f(A)\geq \Gamma$ and $|\mathcal{U}|\geq \mathrm{lt}(\Gamma,\beta, \theta)$, then we have $\mathbb{P}\{\bar{f}(\mathcal{U},A)< (1-\beta) \Gamma \}\leq \theta$
\label{lma:utltbound}
\end{lemma}


\begin{lemma}
Let $\mathcal{U}$ be any set of RR-sets satisfying $|\mathcal{U}|\geq \mathrm{ut}(\Gamma,\beta,\frac{\theta}{D(\Gamma)})$. We have $$\mathbb{P}\{\exists A\in \mathcal{Q}[\Gamma]: \bar{f}(\mathcal{U},A)\geq (1+\beta)\Gamma\}\leq \theta.$$
\label{lma:boundingoursolution}
\end{lemma}

Note that $\mathsf{BCGC}$ set $|\mathcal{R}|=T\geq \max\{\mathrm{ut}(W_1), \mathrm{lt}(W_2)\}$ and $\Lambda=(1-{\alpha}+\gamma)\eta$ (lines~\ref{ln:setw1}-\ref{ln:setw2}). According to Lemma~\ref{lma:boundingoursolution}, when $|\mathcal{R}|\geq \mathrm{ut}(W_1)$ and $\mu\geq D(\eta-\alpha\eta)$, we must have
\begin{eqnarray}
\mathbb{P}\{\exists A\in \mathcal{Q}[(1-\alpha)\eta]: \bar{f}(\mathcal{R},A)\geq (1-\alpha+\gamma)\eta\}\leq {\delta}/{2}\label{eqn:ensurefeasible}
\end{eqnarray}
Moreover, when $|\mathcal{R}|\geq \mathrm{ut}(W_2)$, Lemma~\ref{lma:utltbound} gives us:
\begin{eqnarray}
\mathbb{P}\{\bar{f}(\mathcal{R},S_{opt})< (1-\sigma)\eta\}\leq {\delta}/{2} \label{eqn:ensurear}
\end{eqnarray}

Intuitively, Eqn.~(\ref{eqn:ensurefeasible}) ensures that none of the infeasible solutions in $\mathcal{Q}[(1-\alpha)\eta]$ can be returned by $\mathsf{BCGC}$ with high probability, while Eqn.~(\ref{eqn:ensurear}) ensures that $\bar{f}(\mathcal{R},S_{opt})$ must be close to $\eta$. Using (\ref{eqn:ensurefeasible})-(\ref{eqn:ensurear}), we can get:

\begin{theorem}
When $\sigma,\gamma>0$, $\sigma+\gamma<\alpha<1$ and $\mu\geq D(\eta-\alpha\eta)$, $\mathsf{BCGC}$ returns a set $S\subseteq V$ satisfying $f(S)\geq (1-\alpha)\eta$ and $C(S)\leq (1+\ln\frac{1-\sigma}{\alpha-\gamma-\sigma})C(S_{opt})$ with the probability of at least $1-\delta$.
\label{thm:approximationratio}
\end{theorem}
\begin{proof}
Let $B^*$ denote an optimal solution to the optimization problem: $$\min\{C(A)| \bar{f}(\mathcal{R},A)\geq (1-\sigma)\eta \wedge A\subseteq V\}.$$ As $\mathsf{MCA}$ is essentially a deterministic greedy submodular cover algorithm and $(1-\alpha+\gamma)\eta <(1-\sigma)\eta$, we can use Fact~\ref{fct:goyalgreedy} to get $$C(S)\leq (1+\ln\frac{1-\sigma}{\alpha-\gamma-\sigma})C(B^*).$$ Therefore, if $\bar{f}(\mathcal{R},S_{opt})\geq (1-\sigma)\eta$ holds, then we must have $C(B^*)\leq C(S_{opt})$ and hence $$C(S)\leq (1+\ln\frac{1-\sigma}{\alpha-\gamma-\sigma})C(S_{opt}).$$ According to (\ref{eqn:ensurear}), the probability that $\bar{f}(\mathcal{R},S_{opt})\geq(1-\sigma)\eta$ does not hold is at most $\delta/2$. Besides, as $\bar{f}(\mathcal{R},S)\geq (1-\alpha+\gamma)\eta$, the probability that $\mathsf{BCGC}$ returns an infeasible solution in $\mathcal{Q}[(1-\alpha)\eta]$ is no more than $\mathbb{P}\{\exists A\in \mathcal{Q}[(1-\alpha)\eta]: \bar{f}(\mathcal{R},A)\geq (1-\alpha+\gamma)\eta\}$, which is bounded by $\delta/2$ due to (\ref{eqn:ensurefeasible}). The theorem then follows by using the union bound.
\end{proof}

Note that the approximation ratio of $\mathsf{BCGC}$ nearly matches the approximation ratio (i.e., $1+\ln\frac{1}{\alpha}$) shown in Fact~\ref{fct:goyalgreedy}, which is derived under the EV setting. For example, if we set $\sigma=\gamma=\frac{\alpha}{4}$, then the approximation ratio of $\mathsf{BCGC}$ is at most $1+\ln\frac{2}{\alpha}$, which is larger than $1+\ln\frac{1}{\alpha}$ by only $\ln 2$.

$\mathsf{BCGC}$ spends most of its running time on generating RR-sets. According to the setting of $T$, it can be seen that $\mathsf{BCGC}$ generates at most $\mathcal{O}(\frac{n}{\alpha^2}\ln \frac{n}{\eta\delta})$ RR sets (see the proof of Theorem~\ref{thm:timecomplexityBCGC}). Therefore, we can get:
\begin{theorem}
Let $q=\max\{f(v)|v\in V\}$. $\mathsf{BCGC}$ can achieve the performance bound shown in Theorem~\ref{thm:approximationratio} under the expected time complexity of $\mathcal{O}(\frac{m q}{\alpha^2}\ln\frac{n}{(1-\alpha)\delta\eta})$.
\label{thm:timecomplexityBCGC}
\end{theorem}

    \begin{algorithm} [tp!]
    \KwIn{\textcolor{black}{$G=(V,E), \delta, \alpha,\sigma,\gamma, \eta$}}
    \KwOut{A set $S\subseteq V$ satisfying $f(S)\geq (1-\alpha)\eta$ w.h.p.}
  $W_1\leftarrow \langle (1-\alpha)\eta,\frac{\gamma}{1-\alpha},\frac{\delta}{6D(\eta-\alpha\eta)} \rangle$;~$W_2\leftarrow \langle \eta,\sigma,\frac{\delta}{6} \rangle$ \label{ln:setttingTtegc1}\\
  $T\leftarrow \mathsf{SetT}(W_1,W_2)$;~$\theta\leftarrow {\delta}/{3};~\mathcal{R}\leftarrow \emptyset$ \label{ln:setttingTtegc2}\\
  \While{$|\mathcal{R}|\leq T$}{ \label{ln:iterationstart}
    Generate some RR-sets and add them into $\mathcal{R}$ until $|\mathcal{R}|=\min\{T,\lceil \mathrm{lt}(\langle \eta,\sigma,\frac{\theta}{3}\rangle)\rceil\}$ \label{ln:iterategenerate2}\\
    $S\leftarrow \mathsf{MCA}(\mathcal{R},(1-\alpha+\gamma)\eta)$ \label{ln:findsbigger1minusalphaplus}\\
    \lIf{$|\mathcal{R}|=T$}{\Return{$S$} \label{ln:achievesthethreldshold}}
    $(\mathcal{U},\mathrm{Pass})\leftarrow \mathsf{TEST}(S,\frac{\gamma}{2(1-\alpha)}, (1-\alpha)\eta,\textcolor{black}{\frac{2\theta}{3}},T-|\mathcal{R}|)$ \label{ln:calltest}\\
    \lIf{$\mathrm{Pass}=\mathbf{True}$}{\Return{$S$} \label{ln:returnbypass}}
    $\mathcal{R}\leftarrow \mathcal{R}\cup \mathcal{U}$;~$\theta\leftarrow \frac{\theta}{2}$ \label{ln:iterationend}\\
  }
  \Return{$S$}
    \caption{The Trial-and-Error Algorithm for General Costs ($\mathsf{TEGC}$)}
    \label{alg:mca}
  \end{algorithm}

\subsection{A Trial-and-Error Algorithm}
\label{sec:tae}

It can be seen that $\mathsf{BCGC}$ behaves in an ``once-for-all'' manner, i.e, it generates all the RR-sets in one batch, and then finds a solution using the generated RR-sets. In this section, we propose another ``trial-and-error'' algorithm (called $\mathsf{TEGC}$) for the MCSS problem, which runs in iterations and ``lazily'' generates RR-sets when necessary.

%
%
%

 More specifically, $\mathsf{TEGC}$ runs in iterations with the input variables satisfying $\sigma+\gamma<\alpha<1$. In each iteration (lines~\ref{ln:iterationstart}-\ref{ln:iterationend}), it first generates a set $\mathcal{R}$ of RR-sets (line~\ref{ln:iterategenerate2}) according to Lemma~\ref{lma:utltbound} to ensure $$\mathbb{P}\{\bar{f}(\mathcal{R},S_{opt})< (1-\sigma)\eta\}\leq {\theta}/{3},$$ 
where the parameter $\theta$ will be explained shortly. Then $\mathsf{TEGC}$ calls $\mathsf{MCA}$ in line~\ref{ln:findsbigger1minusalphaplus} to find an approximation solution $S$ satisfying $$\bar{f}(\mathcal{R},S)\geq (1-\alpha+\gamma)\eta.$$ However, it is possible that $S$ is an infeasible solution in $\mathcal{Q}[(1-\alpha)\eta]$. Therefore, $\mathsf{TEGC}$ calls $\mathsf{TEST}$ to judge whether $f(S)\geq (1-\alpha)\eta$ (line~\ref{ln:calltest}). If $\mathsf{TEST}$ returns $\mathrm{Pass}=\mathbf{True}$, it implies that $S$ is a feasible solution w.h.p., so $\mathsf{TEGC}$ terminates and returns $S$. Otherwise, $\mathsf{TEGC}$ enters into another iteration and adds more RR-sets into $\mathcal{R}$. This ``trial and error'' process repeats until $|\mathcal{R}|$ achieves a predefined threshold $T$ (line~\ref{ln:achievesthethreldshold}). As $T=\mathcal{O}(\frac{n}{\alpha^2}\ln \frac{n}{\eta\delta})$ is set similarly with that in $\mathsf{BCGC}$ (lines~\ref{ln:setttingTtegc1}-\ref{ln:setttingTtegc2}), $\mathsf{TEGC}$ is guaranteed to find a feasible solution with high probability. 


The parameter $\theta$ in $\mathsf{TEGC}$ is roughly explained as follows. Intuitively, $\theta$ indicates the total probability of the ``bad events'' (e.g., $\{\bar{f}(\mathcal{R},S_{opt})< (1-\sigma)\eta\}$) happen in any iteration of $\mathsf{TEGC}$ when $|\mathcal{R}|<T$. In the first iteration, we set $\theta=\frac{\delta}{3}$ and $\theta$ is decreased by a factor $2$ in every subsequent iteration. We also constrain the probability that the bad events happen when $|\mathcal{R}|=T$ by $\frac{\delta}{3}$ (lines~\ref{ln:setttingTtegc1}-\ref{ln:setttingTtegc2}). Using the union bound, the total probability that $\mathsf{TEGC}$ returns a ``bad'' solution conflicting our performance bounds is upper-bounded by $\frac{\delta}{3}+\sum_{i=0}^\infty \frac{\delta}{3\cdot 2^i}\leq \delta$.

\textcolor{black}{By similar reasoning with that in Theorem~\ref{thm:approximationratio}, we can prove that $\mathsf{TEGC}$ has the same approximation ratio as $\mathsf{BCGC}$:}
\begin{theorem}
When $\sigma,\gamma>0$ and $\sigma+\gamma<\alpha<1$, $\mathsf{TEGC}$ returns a set $S\subseteq V$ satisfying $f(S)\geq (1-\alpha)\eta$ and $C(S)\leq (1+\ln\frac{1-\sigma}{\alpha-\gamma-\sigma})C(S_{opt})$ with the probability of at least $1-\delta$.
\label{thm:artegc}
\end{theorem}

\noindent \textit{\textbf{The Design of function}} $\mathsf{TEST}$: Next, we explain how the function $\mathsf{TEST}(A,\kappa, \Gamma,\beta,L)$ is implemented. $\mathsf{TEST}$ maintains three threshold values $\ell, M, L$. When $L\leq M$, it simply generates $L$ RR-sets and returns them (lines~\ref{ln:returnLrrsetsbegin}-\ref{ln:returnLrrsetsend}). Otherwise, it keeps generating RR-sets $U_1,U_2,\cdots$ until either $\sum_{i=1}^j X(U_i,A)\geq \ell$ or $|\mathcal{U}|= M$, where $\mathcal{U}$ is the set of generated RR-sets. Note that $\mathsf{TEST}$ is called with $L=T-|\mathcal{R}|$, which implies that the total number of generated RR-sets in $\mathsf{TEGC}$ never exceeds $T$. 

Intuitively, if $f(A)$ is very large, then $X(U_i,A)$ ($\forall i$) must have a high probability to be $1$, so there is a high probability that $\sum_{i=1}^M X(U_i,A)\geq \ell$ and hence $\mathsf{TEST}$ returns $\mathrm{Pass}=\mathbf{True}$. Conversely, if $f(A)$ is very small, then there is a high probability that $\sum_{i=1}^M X(U_i,A)<\ell$ and hence $\mathsf{TEST}$ returns $\mathrm{Pass}=\mathbf{False}$. By setting the values of $\ell$ and $M$ (line~\ref{ln:settingml}) based on the Chernoff bounds, we get the following theorem:


\begin{theorem}
For any $A\subseteq V$, if $f(A)\geq (1+\kappa)\Gamma$ and $L> M$, then the probability that $\mathsf{TEST}(A,\kappa, \Gamma,\beta,L)$ returns $\mathrm{Pass}=\mathbf{True}$ is at least $1-{\beta}/{2}$; if $f(A)<\Gamma$ and $L> M$, then the probability that $\mathsf{TEST}(A,\kappa, \Gamma,\beta,L)$ returns $\mathrm{Pass}=\mathbf{False}$ is at least $1-{\beta}/{2}$.
\label{thm:betais1}
\end{theorem}

\textcolor{black}{Note that $\mathsf{TEST}$ is called by $\mathsf{TEGC}$ with $A=S$, $\Gamma= (1-\alpha)\eta$ and $\kappa=\frac{\gamma}{2(1-\alpha)}$.} So Theorem~\ref{thm:betais1} implies that $\mathsf{TEGC}$ always returns a feasible solution with high probability. When $(1-\alpha)\eta \leq f(S) < (1-\alpha+\gamma/2)\eta$, it is possible that $\mathsf{TEST}$ returns $\mathrm{Pass}=\mathbf{False}$, but this does not harm the the correctness of $\mathsf{TEGC}$ and only results in more iterations; moreover, the probability for this event to happen can be very small as $\gamma$ is usually small.

      \begin{algorithm} [tp!]
  $\ell \leftarrow \left\lceil\frac{2(1+\kappa)\Gamma}{(2+\kappa)n}+\frac{8(3+2\kappa)(1+\kappa)}{3\kappa^2}\ln\frac{2}{\beta}\right\rceil$; \label{ln:settingml}
  $M\leftarrow \left\lfloor \frac{(2+\kappa)n\ell}{2(1+\kappa)\Gamma}\right\rfloor$\\
  $\mathcal{U}\leftarrow \emptyset;~\mathrm{Pass}\leftarrow \mathbf{False};~Z_0\leftarrow 0;$\\
  \If{$L\leq M$}{\label{ln:returnLrrsetsbegin}
    Generate $L$ RR-sets and add them into $\mathcal{U}$\\
    \Return{$(\mathcal{U}, \mathrm{Pass})$ \label{ln:returnLrrsetsend}}
    }
  \For{$j=1$ \KwTo $M$}{
    Generate an RR-set $U_j$ and add it into $\mathcal{U}$\\
    $Z_j\leftarrow Z_{j-1}+ \min\{|A\cap U_{j}|,1\}$\\
    \If{$Z_j= \ell$}{
        $\mathrm{Pass}\leftarrow \mathbf{True}$;~\Return{$(\mathcal{U},\mathrm{Pass})$}
        }
  }
  \Return{$(\mathcal{U}, \mathrm{Pass})$}
    \caption{$\mathsf{TEST}(A,\kappa, \Gamma,\beta,L)$}
    \label{alg:mca}
  \end{algorithm}

\subsection{Theoretical Comparisons for the BA Algorithms}
\label{sec:compareundergc}

We compare the theoretical performance of $\mathsf{BCGC}$ and $\mathsf{TEGC}$ with the state-of-the-art algorithms as follows.
\subsubsection{Comparing with Goyal et.al.'s Results~\cite{Goyal2013}}

To the best of our knowledge, \textbf{the only prior algorithm with a provable performance bound for the MCSS problem under the GC+NV setting is the one proposed in~\cite{Goyal2013}}, which is based on the $\mathsf{GreSSC}$ algorithm. We quote the result of~\cite{Goyal2013} in Fact~\ref{fct:goyalsresult}:

\begin{fact}[\cite{Goyal2013}]
Let $S'=\mathsf{GreSSC}((1-\alpha)\eta,\beta,0,0)$. Then we have $f(S')\geq (1-\alpha)\eta$ and $C(S')\leq (1+\phi)(1+\ln\frac{1}{\alpha})C(S_{opt})$, where $\phi$ and $\beta$ satisfy:
\begin{eqnarray}
\frac{\beta}{1-(1-\beta)(1-{1}/{C(S_{opt})})}=\frac{(1/\alpha)^{\phi}-1}{(1/\alpha)^{\phi+1}-1}
\end{eqnarray}
\label{fct:goyalsresult}
\end{fact}

However, Goyal \textit{et.al.}~\cite{Goyal2013} did not mention how to implement $\mathsf{GreSSC}((1-\alpha)\eta,\beta,0,0)$. A crucial problem is that they require $\hat{f}(A')\leq {f}(A')~(\forall A'\subseteq V)$ in the implementation of $\mathsf{GreSSC}$ (i.e., $\gamma_2=0$); otherwise their proof for the approximation ratio in Fact~\ref{fct:goyalsresult} does not hold. However, as computing $f(A')$ is \#P-hard and $\hat{f}(A')$ in $\mathsf{GreSSC}$ is computed by monte-carlo simulations, $\gamma_2=0$ actually implies that infinite times of monte-carlo simulations should be conducted to compute $\hat{f}(A')$ according to the Chernoff bounds~\cite{FanL2006}. Due to this reason, we think that the approximation ratio shown in Fact~\ref{fct:goyalsresult} is mainly valuable on the theoretical side and is not likely to be implemented in polynomial time.

\subsubsection{Comparing with Kuhnle et.al.'s Results~\cite{Kuhnle2017}}

Very recently, Kuhnle \textit{et.al.}~\cite{Kuhnle2017} have proposed some elegant bi-criteria approximation algorithms for the MCSS problem, but only under the UC setting. We quote their results below:

\begin{fact}[\cite{Kuhnle2017}]
Define the CEA assumption as follows: for any $A\subseteq V$ such that $f(A)<\eta$, there always exists a node $u$ such that $f(A\cup \{u\})-f(A)\geq 1$. Under the UC setting, the STAB algorithms can find a set $S'\subseteq V$ satisfying $|S'|\leq (1+2\rho \eta +\log \eta)|S_{opt}|$ and 1) or 2) listed below with the probability of at least $1-\iota n^3$. The STAB-C1 algorithm has $\mathcal{O}((n+m)|S'|\log\frac{1}{\iota n^2}/\rho^2)$ time complexity. The STAB-C2 algorithm has $\mathcal{O}((n+m)|S'|^2\log |S'|\log\frac{1}{\iota n^2}/\rho^2)$ time complexity. 
\begin{enumerate}
\item $f(S')\geq (1-\rho)\eta$ when the CEA assumption holds.
\item $f(S')\geq \eta -(1+\varepsilon)|S_{opt}|$ without the CEA assumption.
\end{enumerate}
\label{fct:kuhapproxalg}
\end{fact}


Note that Fact~\ref{fct:kuhapproxalg} depends on the CEA assumption, and there is a large gap between $\eta$ and $f(S')$ (at least $|S_{opt}|$) if without the CEA assumption. In contrast, the performance ratios of $\mathsf{BCGC}$ or $\mathsf{TEGC}$ do not require any special properties of the network and we can guarantee $f(S)\geq (1-\alpha)\eta$ w.h.p. for any $\alpha>0$. Most importantly, Fact~\ref{fct:kuhapproxalg} only holds under the UC setting while our performance bounds shown in Theorems~\ref{thm:approximationratio}-\ref{thm:artegc} hold under the GC setting.

%
%

\section{Approximation Algorithms for Uniform Costs}
\label{sec:approxforuc}


In this section, we propose approximation (AP) algorithms for the MCSS problem under the UC+NV setting. 

\subsection{The AAUC Algorithm}

We first propose an algorithm called $\mathsf{AAUC}$. $\mathsf{AAUC}$ first generates a set $\mathcal{R}$ of $T$ RR-sets, and then calls $\mathsf{MCA}$ to return a set $S$ satisfying $\bar{f}(\mathcal{R},S)\geq \Lambda=(1+\tau)\eta$, where $\tau$ is a variable in $(0,1)$. 
Although $\mathsf{AAUC}$ looks similar to $\mathsf{BCGC}$, its key idea (i.e., how to set the parameters $T$ and $\Lambda$) is very different from that in $\mathsf{BCGC}$, which is explained as follows.

Recall that $\mathsf{MCA}$ is a greedy algorithm. Suppose that $\mathsf{MCA}$ sequentially selects $x_1,x_2,\cdots,x_n$ after it is called by $\mathsf{AAUC}$. Let $B_i(\forall 1\leq i\leq n)$ denote $\{x_1,\cdots,x_i\}$. Let $s=\min\{i|\bar{f}(\mathcal{R},B_i)\geq (1-\tau)\eta\}$ and $t=\min\{i|\bar{f}(\mathcal{R},B_i)\geq (1+\tau)\eta\}$.
As each node has a cost $1$ and $\bar{f}(\mathcal{R},\cdot)$ is submodular, we can use the submodular functions' properties to prove:
\begin{lemma}
For any $\tau\in (0,1)$, we have $|B_{s}|\leq  \lceil \ln\frac{n\eta}{n-\eta}\rceil |D^*|$ $+1$, where $D^*=\arg\min_{A\subseteq V\wedge \bar{f}(\mathcal{R},A)\geq (1-\tau)\eta}|A|$.
\label{lma:thearoftminus1}
\end{lemma}
\noindent Clearly, when $\tau$ is sufficiently small, $s$ and $t$ must be very close. Indeed, we can prove:
\begin{lemma}
When $\tau\in (0,\frac{n-\eta}{2n\eta+\eta}]$, we have $|B_{t}|\leq |B_{s}|+1$.
\label{lma:atminus1issmall}
\end{lemma}

\noindent Besides, according to line~\ref{ln:setT} and Lemmas~\ref{lma:utltbound}-\ref{lma:boundingoursolution}, we know that the set $\mathcal{R}$ generated by $\mathsf{AAUC}$ satisfies
\begin{eqnarray}
&&\mathbb{P}\{\exists A\in \mathcal{Q}[\eta]: \bar{f}(\mathcal{R},A)\geq (1+\tau)\eta\}\leq {\delta}/{2} \label{eqn:aaucfeasible}\\
&&\mathbb{P}\{\bar{f}(\mathcal{R},S_{opt})< (1-\tau)\eta\}\leq {\delta}/{2} \label{eqn:aaucoptlarge}
\end{eqnarray}
\textcolor{black}{for any $\tau\in (0,1)$ and $\mu\geq D(\eta)$.} This implies that $S$ is a feasible solution (i.e., $S\notin \mathcal{Q}[\eta]$) and $|D^*|\leq |S_{opt}|$ with high probability. Moreover, according to Lemmas~\ref{lma:thearoftminus1}-\ref{lma:atminus1issmall}, we have $|S|=|B_{t}|\leq |B_{s}|+1$ and $|B_s|\leq \lceil \ln\frac{n\eta}{n-\eta}\rceil |D^*|+1$. Combining all these results gives us:

    \begin{algorithm} [tp!]
    \KwIn{\textcolor{black}{$ G=(V,E), \eta,\delta, \tau, \mu$}}
    \KwOut{A set $S\subseteq V$ satisfying $f(S)\geq \eta$ w.h.p.}
    %
  $W_1\leftarrow \langle \eta, \tau,\frac{\delta}{2\mu} \rangle;~W_2\leftarrow \langle \eta, \tau,\frac{\delta}{2} \rangle$;~$T\leftarrow \mathsf{SetT}(W_1,W_2)$ \label{ln:setT}\\
  Generate a set $\mathcal{R}$ of RR-sets satisfying $|\mathcal{R}|=T$ \\
  $S\leftarrow \mathsf{MCA}(\mathcal{R},(1+{\tau})\eta)$;~\Return{$S$} \label{ln:return1plustau}
    \caption{Approximation Algorithm for Uniform Costs ($\mathsf{AAUC}$)}
    \label{alg:aaucbasic}
  \end{algorithm}

\begin{theorem}
When $\tau\in (0,\frac{n-\eta}{2n\eta+\eta}]$ and $\mu\geq D(\eta)$, $\mathsf{AAUC}$ returns a set $S\subseteq V$ such that $f(S)\geq \eta$ and $|S|\leq \lceil \ln\frac{n\eta}{n-\eta}\rceil |{S}_{opt}|+2$ with the probability of at least $1-\delta$.
\label{thm:arofaauc}
\end{theorem}
\begin{proof} 
Note that $\bar{f}(\mathcal{R},S)\geq (1+\tau)\eta$ according to line~\ref{ln:return1plustau} of $\mathsf{AAUC}$. Therefore, when $\mu\geq D(\eta)$, the probability that $S$ is an infeasible solution in $\mathcal{Q}[\eta]$ must be no more than $\delta/2$ according to line~\ref{ln:setT} of $\mathsf{AAUC}$ and eqn.~(\ref{eqn:aaucfeasible}).

Let $D^*=\arg\min_{A\subseteq V\wedge \bar{f}(\mathcal{R},A)\geq (1-\tau)\eta}|A|$. When $\tau\in (0,\frac{n-\eta}{2n\eta+\eta}]$, we must have $|S|=|B_t|\leq \lceil \ln\frac{n\eta}{n-\eta}\rceil |D^*|+2$ according to Lemmas~\ref{lma:thearoftminus1}-\ref{lma:atminus1issmall}. Therefore, we have 
\begin{eqnarray}
&&\mathbb{P}\{|S|> \lceil \ln\frac{n\eta}{n-\eta}\rceil |{S}_{opt}|+2\}\leq \mathbb{P}\{|D^*|>|{S}_{opt}|\}\nonumber\\
&\leq& \mathbb{P}\{\bar{f}(\mathcal{R},S_{opt})< (1-\tau)\eta\}\leq {\delta}/{2} \label{eqn:lessthandeltadevide2}
\end{eqnarray}
where (\ref{eqn:lessthandeltadevide2}) is due to the definition of $D^*$ and equation~(\ref{eqn:aaucoptlarge}). The theorem then follows by using the union bound.
\end{proof}

By very similar reasoning with that in Theorem~\ref{thm:timecomplexityBCGC}, we can also prove the time complexity of $\mathsf{AAUC}$ as follows:

\begin{theorem}
Let $q=\max\{f(v)|v\in V\}$. $\mathsf{AAUC}$ can achieve the performance bounds shown Theorem~\ref{thm:arofaauc} under the expected time complexity of $\mathcal{O}(\frac{m q}{\varrho^2}\ln\frac{n}{\delta\eta})$ where $\varrho=\frac{n-\eta}{2n\eta+\eta}$.
\label{thm:timeofaauc}
\end{theorem}



\subsection{An Adaptive Trial-and-Error Algorithm}

It can be seen from Theorem~\ref{thm:timeofaauc} that the running time of $\mathsf{AAUC}$ is inversely proportional to $\varrho$, which can be a small number. To address this problem, we propose an adaptive trial-and-error algorithm called $\mathsf{ATEUC}$, shown in Algorithm~\ref{alg:ateuc}. $\mathsf{ATEUC}$ uses a dynamic parameter $\alpha$ for generating RR-sets, and adaptively changes the value of $\alpha$ until a satisfying approximate solution is found.



More specifically, $\mathsf{ATEUC}$ first determines a threshold $T$ using $\varrho$ (lines~\ref{ln:ateucsett1}-\ref{ln:ateucsett2}), then $\mathsf{ATEUC}$ calls a function $\mathsf{Shrink}$ with any $\alpha$ that is larger than $\varrho$. Note that the performance bound of $\mathsf{ATEUC}$ (see Theorem~\ref{thm:arofateuc}) does not depend on the value of $\alpha$, and setting $\alpha>\varrho$ is only for reducing the number of generated RR-sets.  
In each iteration, $\mathsf{Shrink}$ first generates a set $\mathcal{R}$ of RR-sets in a similar way with that in $\mathsf{TEGC}$ (line~\ref{ln:taegenrrset2}) to ensure $\mathbb{P}\{\bar{f}(\mathcal{R},S_{opt})< (1-\alpha)\eta\}\leq {\theta}/{3}$, where $\theta$ is set by a similar way with that in $\mathsf{TEGC}$.
 After that, $\mathsf{Shrink}$ calls $\mathsf{MCA}$ to find $S_1\subseteq S_2\subseteq V$ such that $\bar{f}(\mathcal{R}, S_1)\geq (1-\alpha)\eta$ and $\bar{f}(\mathcal{R}, S_2)\geq (1+\alpha)\eta$.  If $|S_2|> 2|S_1|$, $\mathsf{Shrink}$ decreases $\alpha$ and $\theta$ and enters the next iteration (lines~\ref{ln:decreasethetaandalpha}). If $|S_2|\leq 2|S_1|$, it implies that $\alpha$ is small enough, so $\mathsf{Shrink}$ calls $\mathsf{TEST}$ to judge whether $f(S_2)\geq \eta$ (line~\ref{ln:ateuccalltest}). If $\mathsf{TEST}$ returns $\mathrm{Pass}=\mathbf{True}$, then $\mathsf{Shrink}$ returns $S_2$ as the solution (line~\ref{ln:returns2withrlessT}), otherwise it decreases the value of $\theta$ and enters the next iteration (line~\ref{ln:decreasetheta}).

    \begin{algorithm} [tp!]
    \KwIn{\textcolor{black}{$G=(V,E), \delta, \eta, \alpha$}}
    \KwOut{A set $S\subseteq V$ satisfying $f(S)\geq \eta$ w.h.p.
    }
  $\varrho\leftarrow \frac{|V|-\eta}{2|V|\eta+\eta}$;~$W_1\leftarrow \langle \eta, \varrho,\frac{\delta}{6D(\eta)} \rangle$; \label{ln:ateucsett1}\\
  $W_2\leftarrow \langle \eta, \varrho,\frac{\delta}{6} \rangle$;~$T\leftarrow \mathsf{SetT}(W_1,W_2)$; \label{ln:ateucsett2}\\
  $(\mathcal{R},S)\leftarrow \mathsf{Shrink}(T,\eta,\alpha,\delta,\varrho)$;~\Return{$S$}
  \hrule
  \textbf{Function} $\mathsf{Shrink}(T, \eta,\alpha,\delta,\varrho)$\\
  $\theta\leftarrow {\delta}/{3};~\mathcal{R}\leftarrow \emptyset$\\
  \While{$|\mathcal{R}|\leq T$}{
    Generate some RR-sets and add them into $\mathcal{R}$ until $|\mathcal{R}|=\min\{T,\lceil \mathrm{lt}(\langle \eta,\alpha,\frac{\theta}{3}\rangle)\rceil\}$ \label{ln:taegenrrset2} \\
    \lIf{$|\mathcal{R}|=T$}{$\alpha\leftarrow \varrho$ \label{ln:settingtau}}
    $S_1\leftarrow \mathsf{MCA}(\mathcal{R}, (1-\alpha)\eta)$;~$S_2\leftarrow \mathsf{MCA}(\mathcal{R}, (1+\alpha)\eta)$\\
    \If{$S_2\neq \emptyset$}{
        \If{$|S_2|\leq 2|S_1|$ \label{ln:s2lessthan2s1}}{
            \lIf{$|\mathcal{R}|=T$}{\Return{$(\mathcal{R},S_2)$}\label{ln:returns2withrequalT}}
            $(\mathcal{U},\mathrm{Pass})\leftarrow \mathsf{TEST}(S_2,\frac{1}{2}\alpha, \eta,{\frac{2\theta}{3}},T-|\mathcal{R}|)$ \label{ln:ateuccalltest}\\
            \lIf{$\mathrm{Pass}= \mathbf{True}$}{\Return{$(\mathcal{R},S_2)$} \label{ln:returns2withrlessT}}
            $\mathcal{R}\leftarrow \mathcal{R}\cup \mathcal{U};~\theta\leftarrow {\theta}/{2}$;~\textbf{continue};  \label{ln:decreasetheta}
        }
    }
    $\alpha\leftarrow {\alpha}/{\sqrt{2}};~\theta\leftarrow {\theta}/{2}$ \label{ln:decreasethetaandalpha}\\
  }
  \Return{$(\mathcal{R},\emptyset)$}
    \caption{The Adaptive Trial-and-Error Algorithm for Uniform Costs ($\mathsf{ATEUC}$)}
    \label{alg:ateuc}
  \end{algorithm}

In the case that the number of RR-sets generated in $\mathsf{Shrink}$ reaches $T$, $\mathsf{Shrink}$ returns $S_2$ by setting $\alpha=\varrho$ (lines~\ref{ln:settingtau},\ref{ln:returns2withrequalT}), so we must have $|S_2|\leq |S_1|+1$ and hence $|S_2|\leq \lceil \ln\frac{n\eta}{n-\eta}\rceil |{S}_{opt}|+2$ according to similar reasoning with that of Lemma~\ref{lma:atminus1issmall} and Theorem~\ref{thm:arofaauc}. However, as $\mathsf{Shrink}$ can return $S_2$ when $|S_2|\leq 2|S_1|$, the value of $\alpha$ is probably much larger than $\varrho$ when $\mathsf{Shrink}$ terminates, so less number of RR-sets can be generated. Moreover, by the setting of $\mathcal{R}$ in line~\ref{ln:taegenrrset2}, we always have $\bar{f}(\mathcal{R},S_{opt})\geq (1-\alpha)\eta$ w.h.p. and hence $|S_1|\leq \lceil \ln\frac{n\eta}{n-\eta}\rceil |S_{opt}|+1$ w.h.p. (using similar reasoning with that in Lemma~\ref{lma:thearoftminus1}). Based on these discussions, we can prove the approximation ratio of $\mathsf{ATEUC}$ as follows:



\begin{theorem}
\textcolor{black}{
With the probability of at least $1-\delta$, $\mathsf{ATEUC}$ returns a set $S\subseteq V$ satisfying $f(S)\geq \eta$ and $|S|\leq 2\lceil \ln\frac{n\eta}{n-\eta}\rceil |{S}_{opt}|+2$. }
\label{thm:arofateuc}
\end{theorem}

\begin{figure*}[htb!]
	  \centering
	   \includegraphics[width=0.65\textwidth]{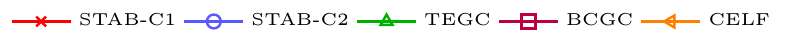}
	   \label{fig:lineTitle_C1_C2_TEGC_BCGC_CELF_alpha02}
	
	  \begin{minipage}[htb!]{1.035\textwidth}
	  	
	    \subfigure[wiki-Vote (RT)]{
	      \includegraphics[width=0.18\textwidth]{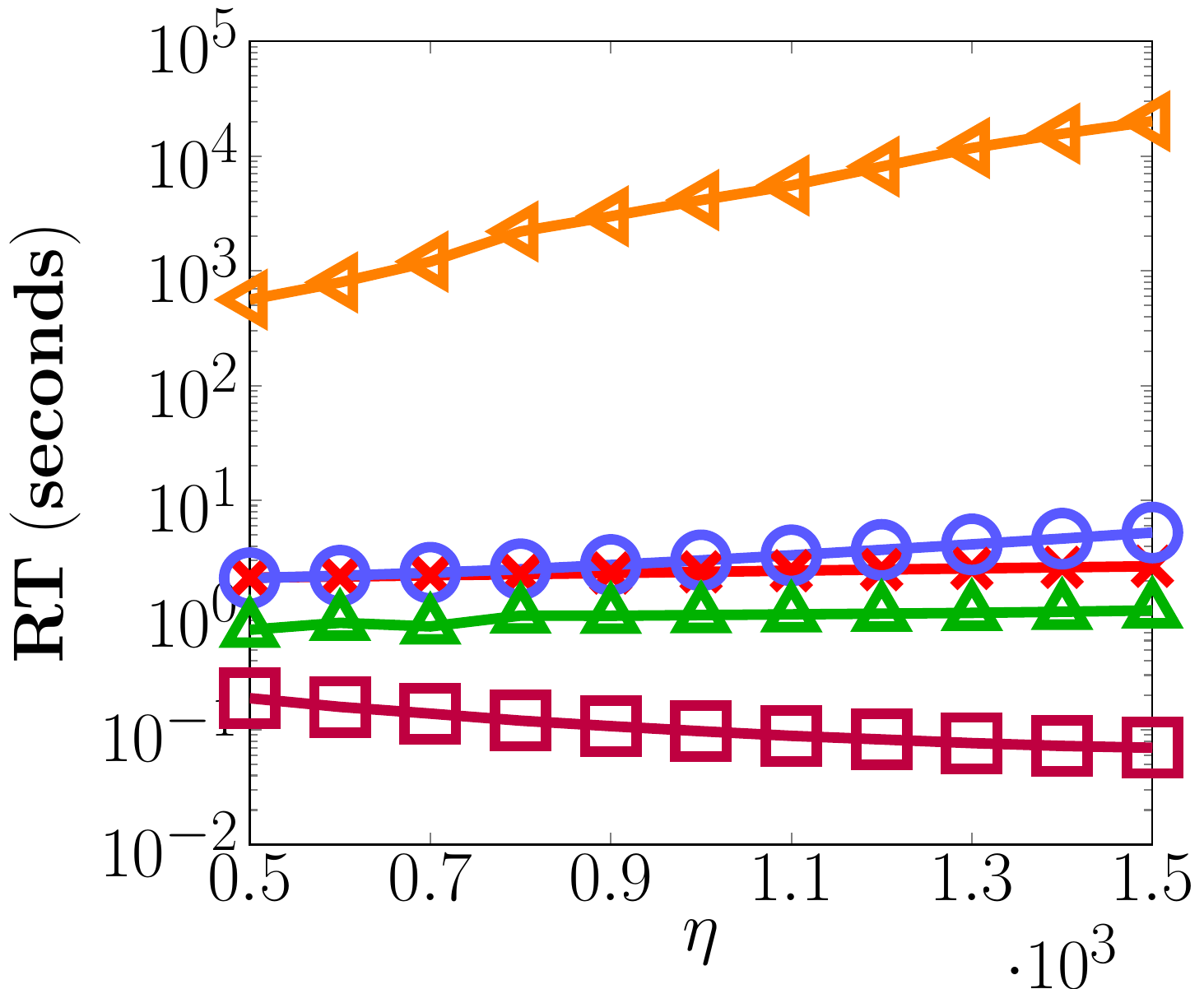}
	      \label{fig:wikiVote_generalCost_RunningTime_alpha02}
	    }
	    \subfigure[Pokec (RT)]{
	      \includegraphics[width=0.18\textwidth]{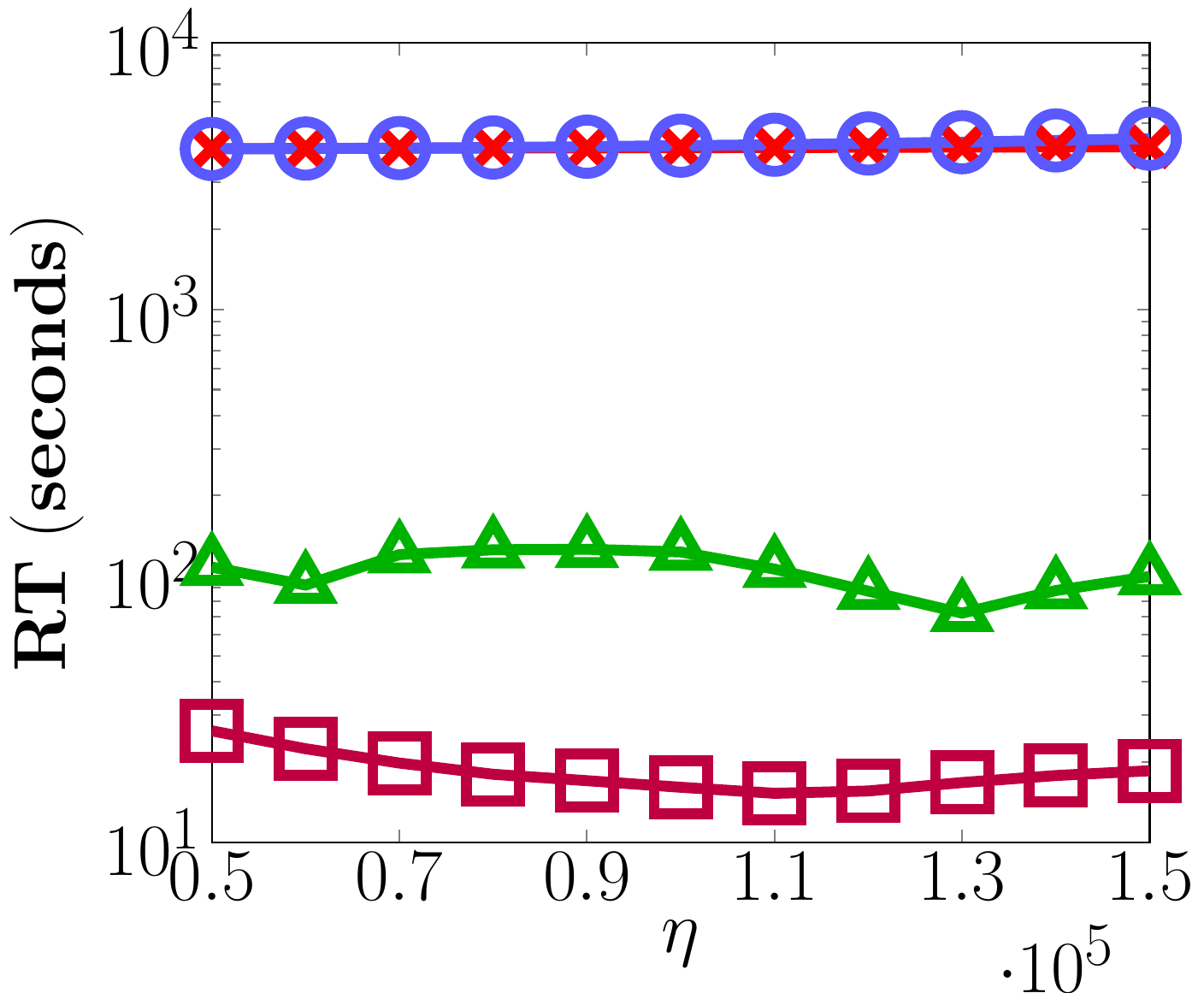}
	      \label{fig:pokec_generalCost_RunningTime_alpha02}
	    }
	     \subfigure[LiveJournal (RT)]{
	      \includegraphics[width=0.18\textwidth]{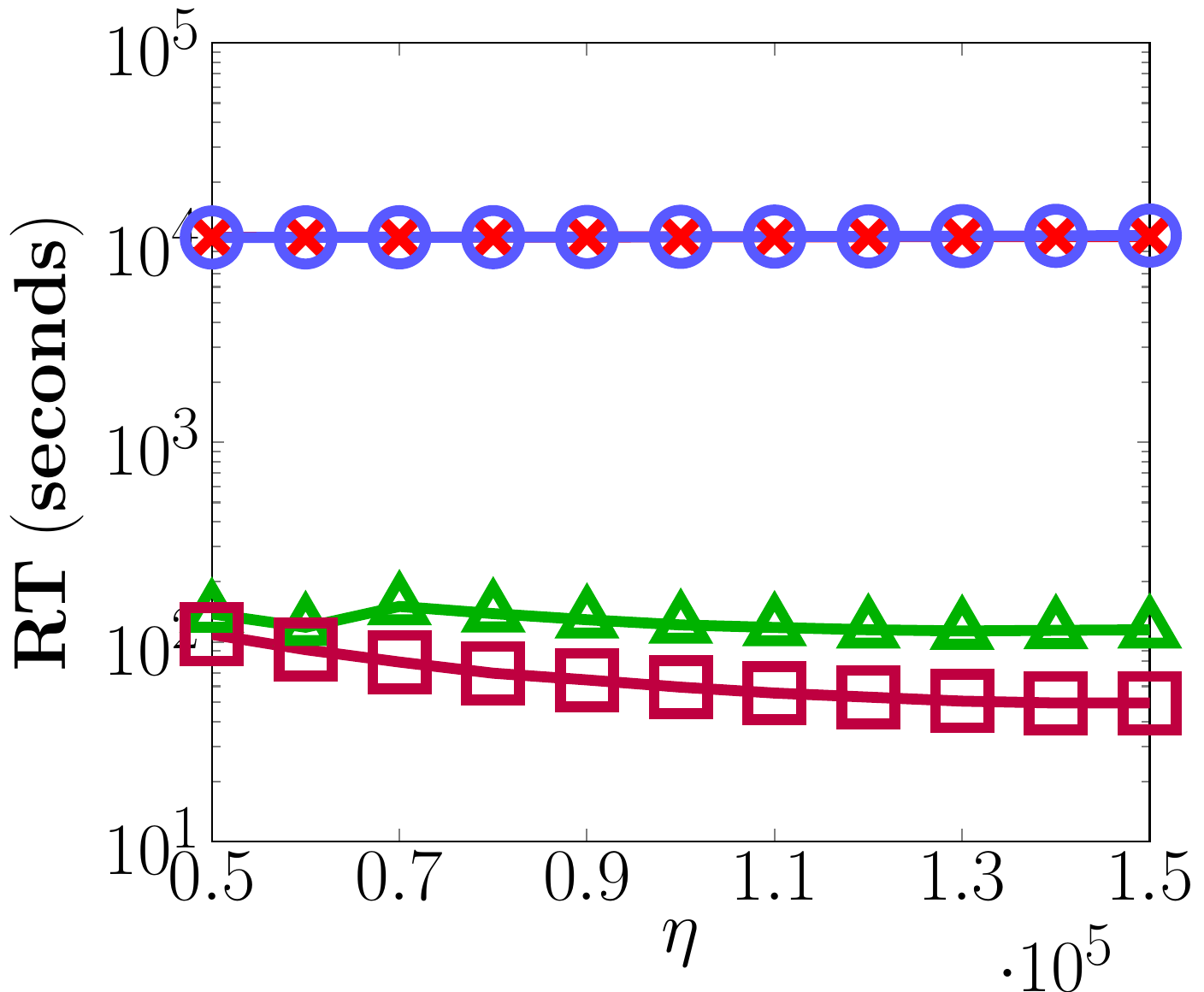}
	      \label{fig:LiveJournal_generalCost_RunningTime_alpha02}
	    }
	    \subfigure[Orkut (RT)]{
	      \includegraphics[width=0.18\textwidth]{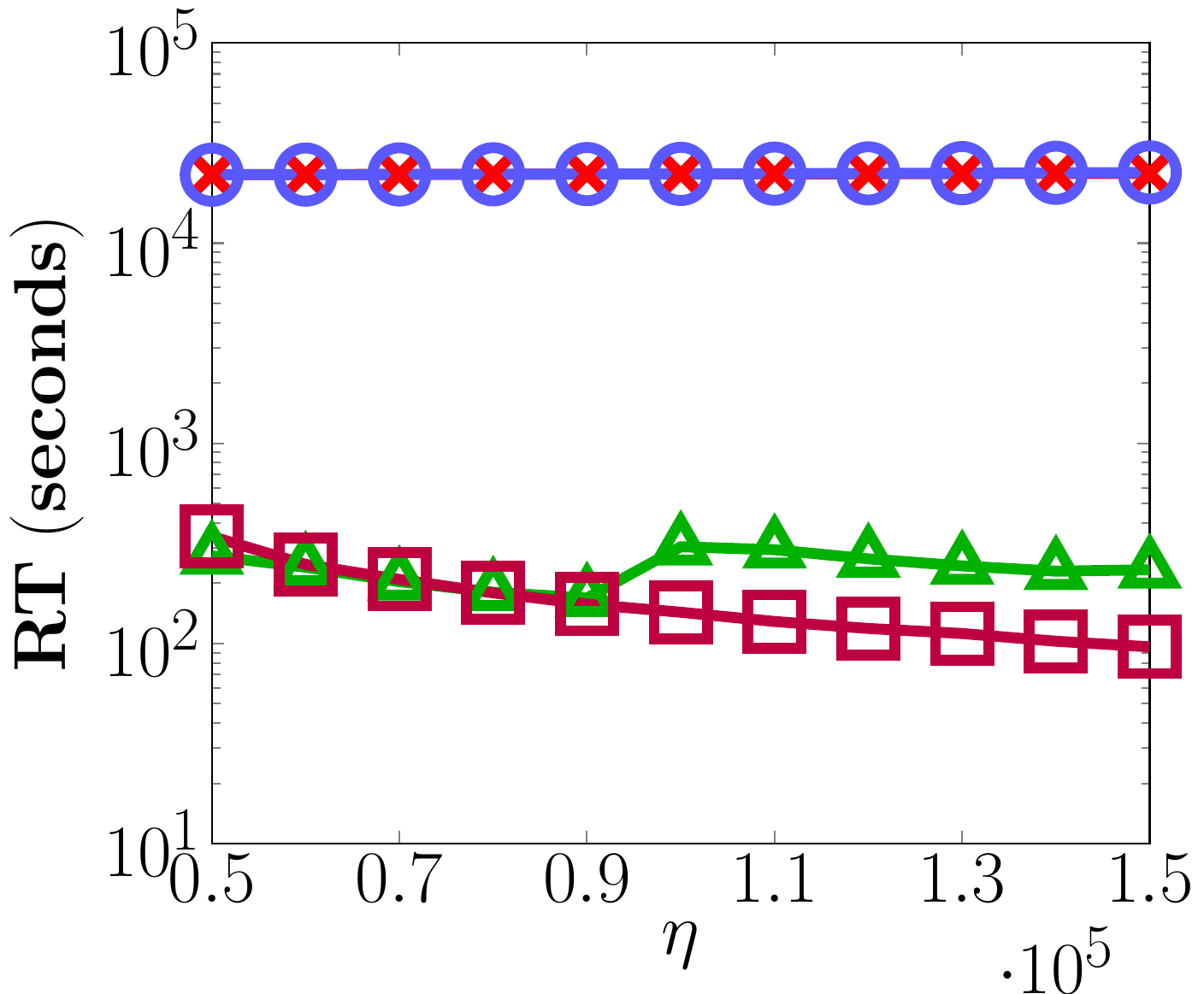}
	      \label{fig:orkut_generalCost_RunningTime_alpha02}
	    }
	    \subfigure[Twitter (RT)]{
	      \includegraphics[width=0.18\textwidth]{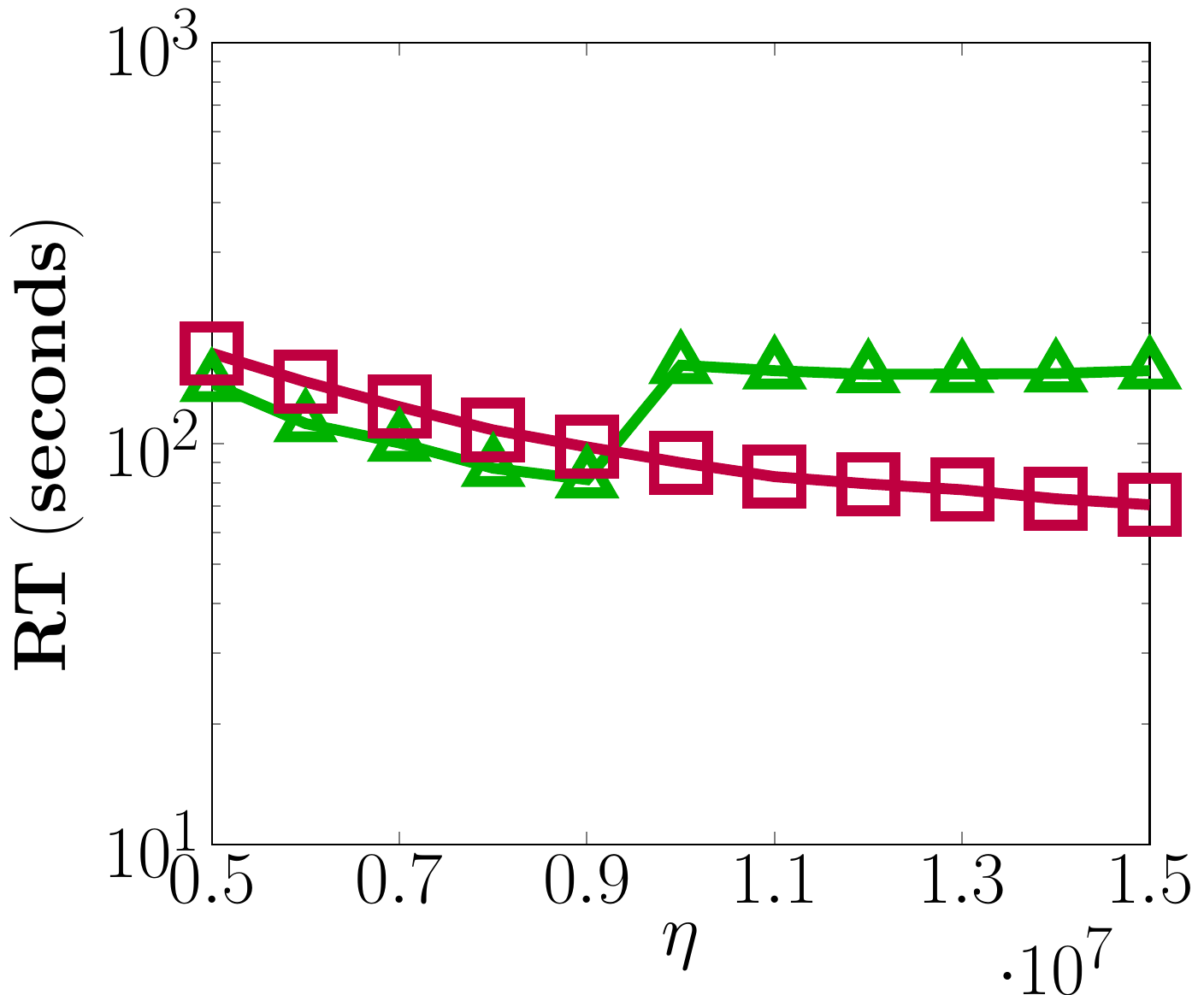}
	      \label{fig:twitter_generalCost_RunningTime_alpha02}
	    } \\
	
	    \subfigure[wiki-Vote (IS)]{
	      \includegraphics[width=0.18\textwidth]{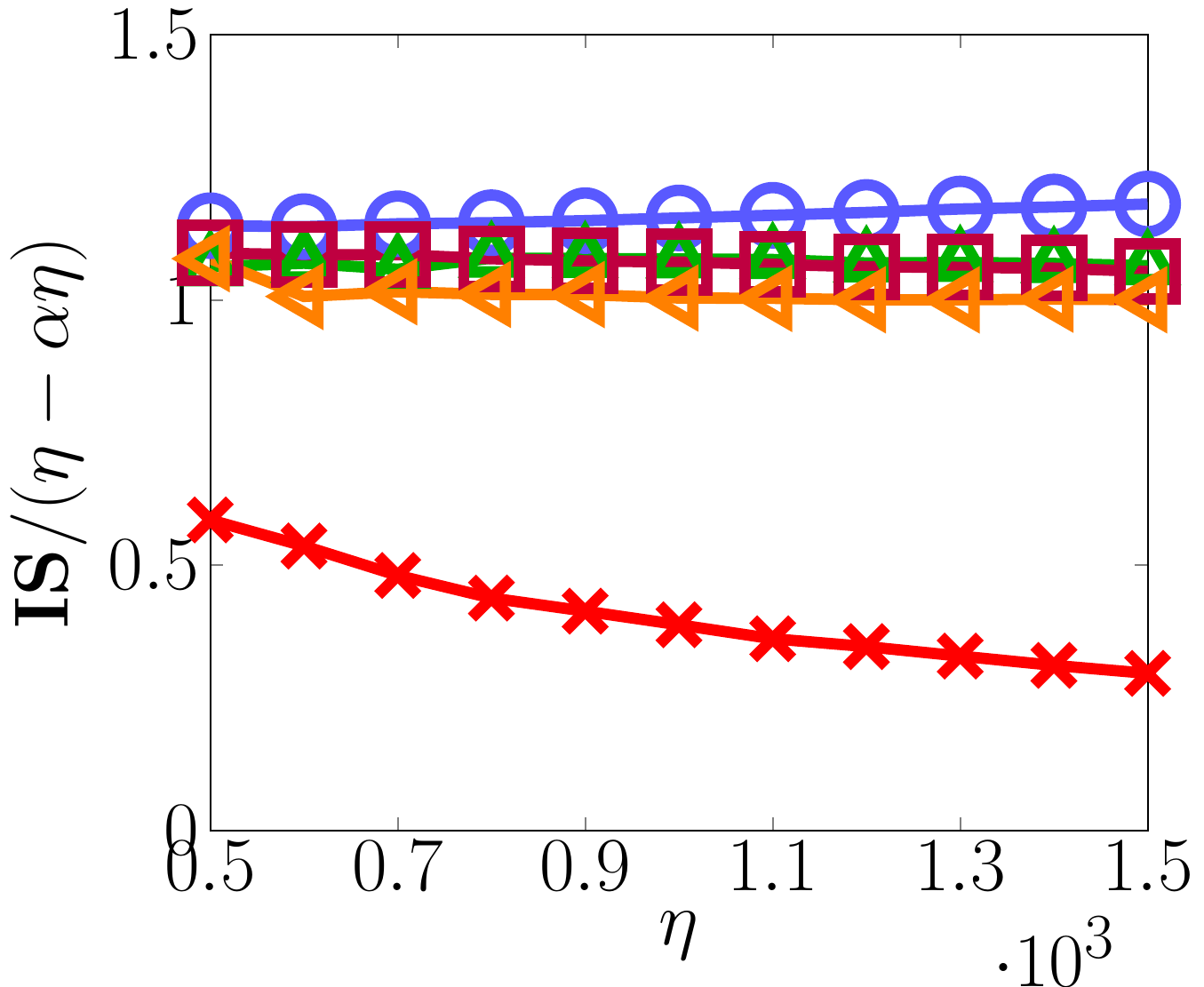}
	      \label{fig:wikiVote_generalCost_Influence_alpha02}
	    }
	    \subfigure[Pokec (IS)]{
	      \includegraphics[width=0.18\textwidth]{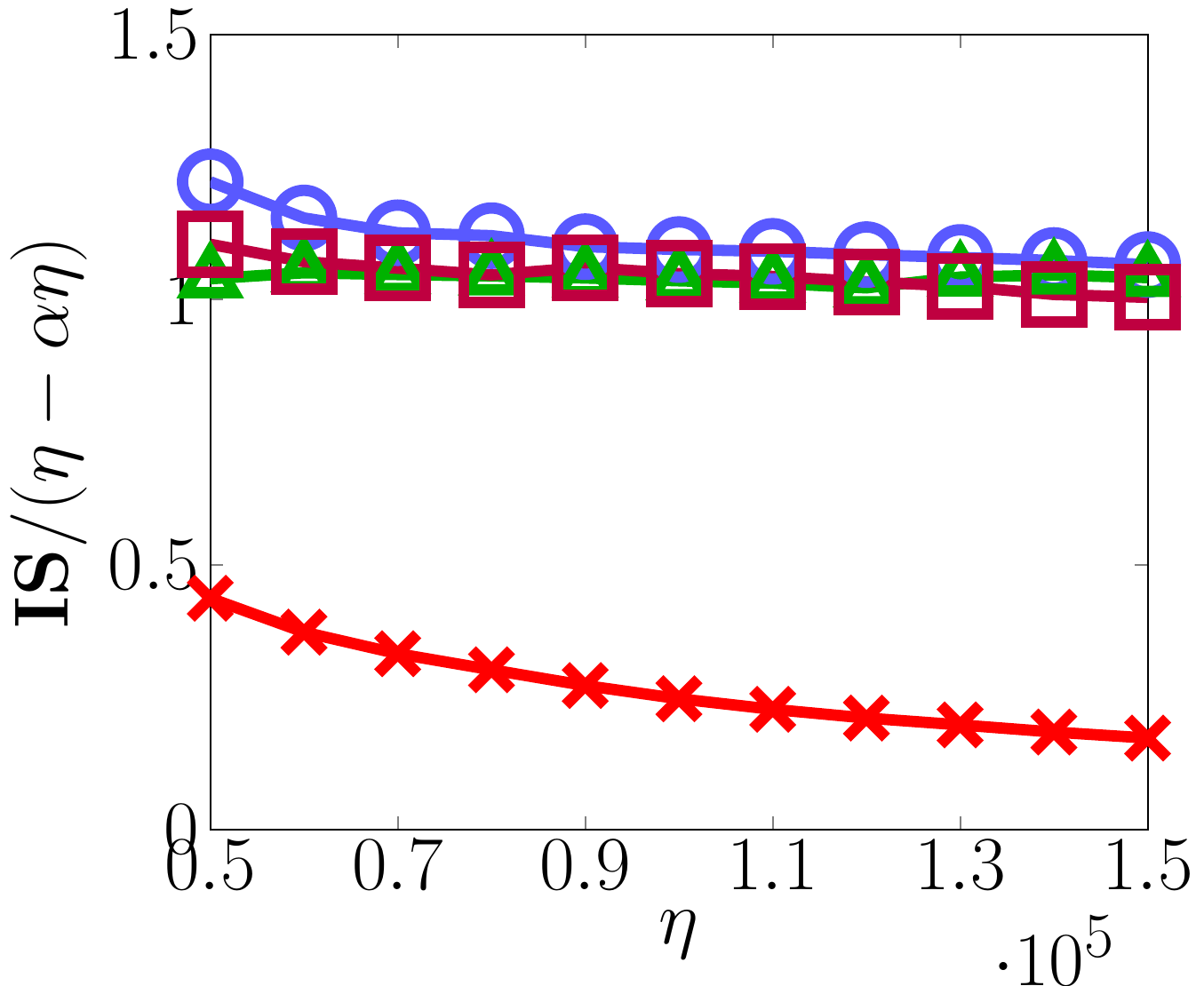}
	      \label{fig:pokec_generalCost_Influence_alpha02}
	    }
	    \subfigure[LiveJournal (IS)]{
	      \includegraphics[width=0.18\textwidth]{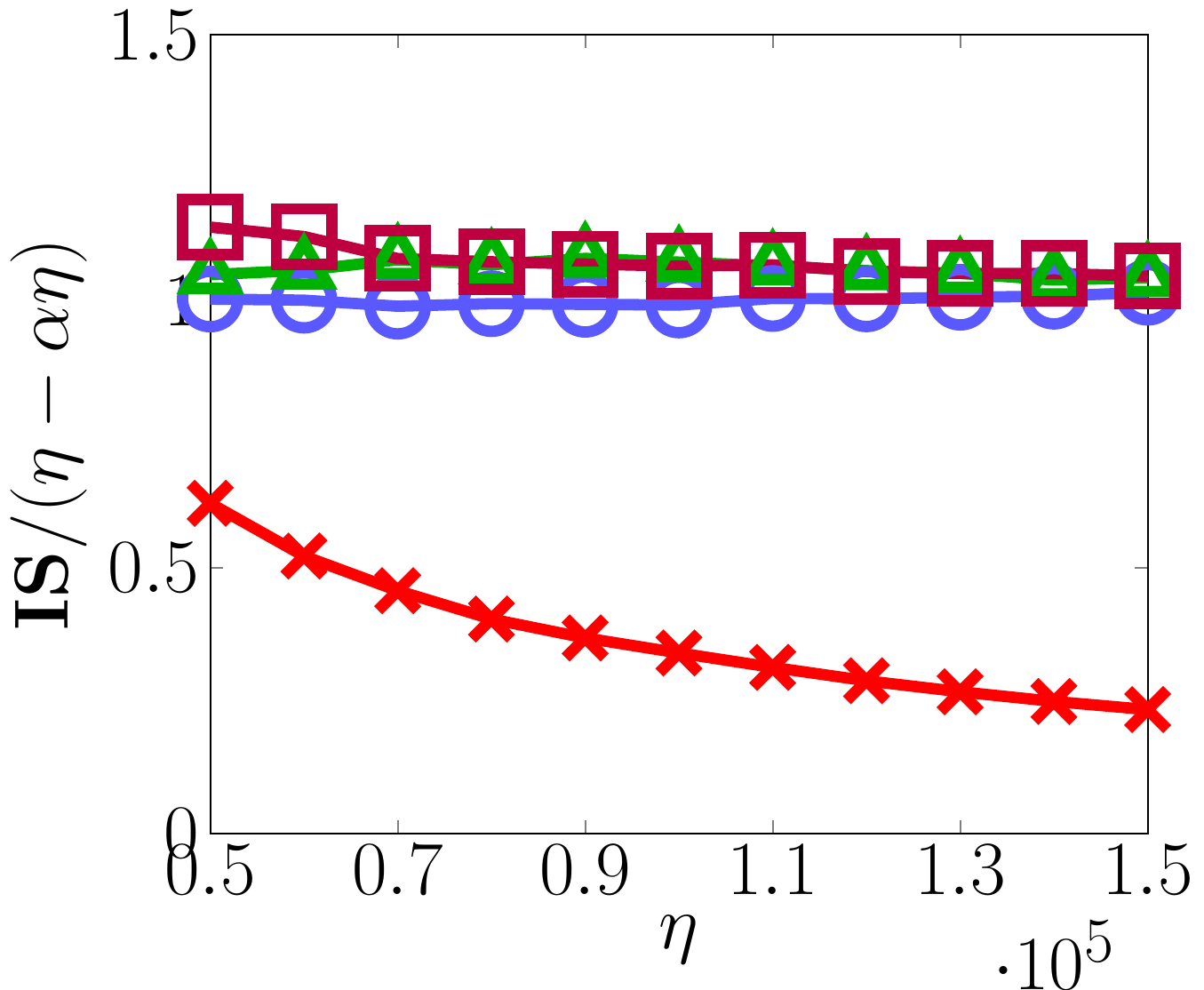}
	      \label{fig:LiveJournal_generalCost_Influence_alpha02}
	    }
	     \subfigure[Orkut (IS)]{
	      \includegraphics[width=0.18\textwidth]{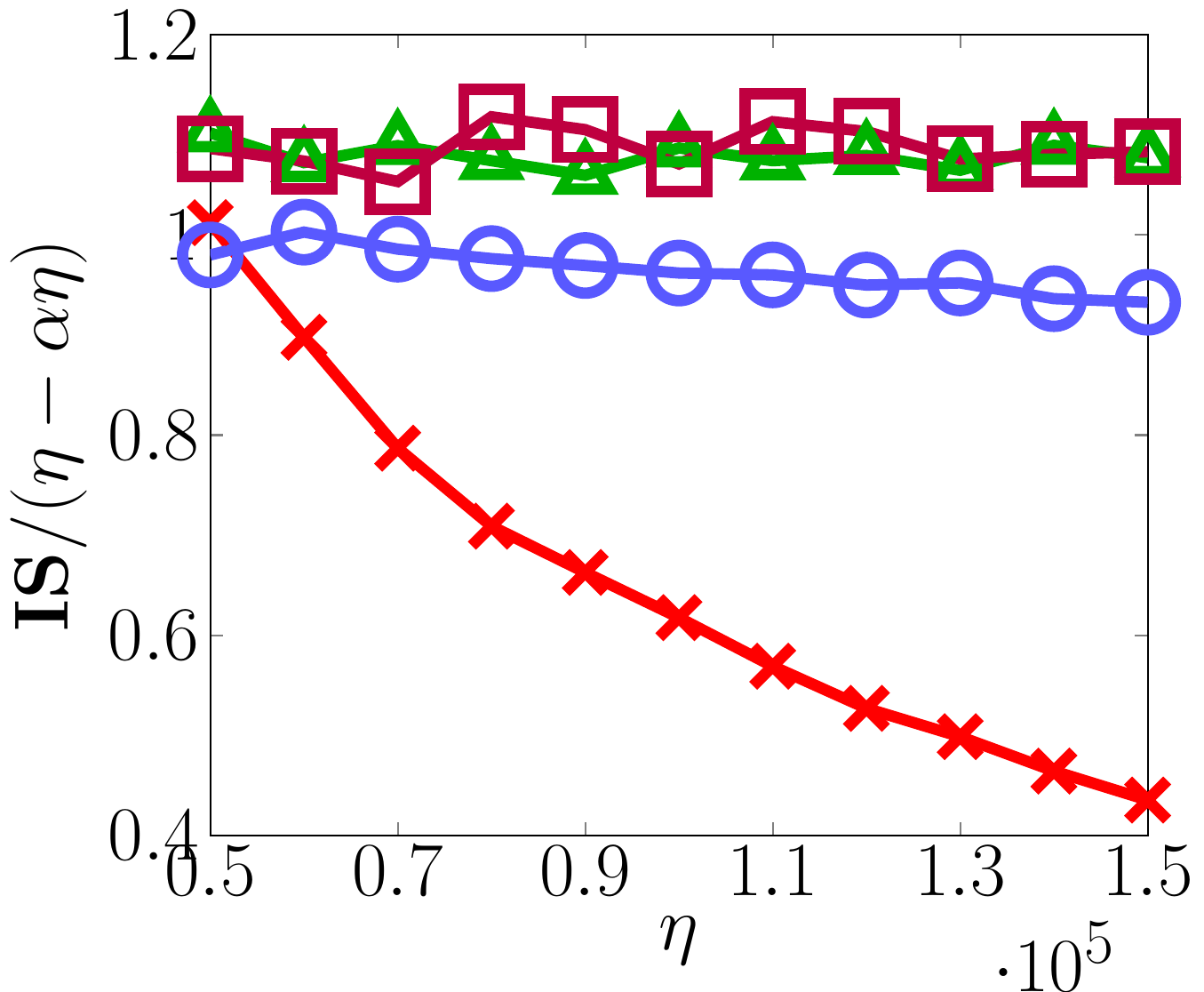}
	      \label{fig:orkut_generalCost_Influence_alpha02}
	    }
	    \subfigure[Twitter (IS)]{
	      \includegraphics[width=0.18\textwidth]{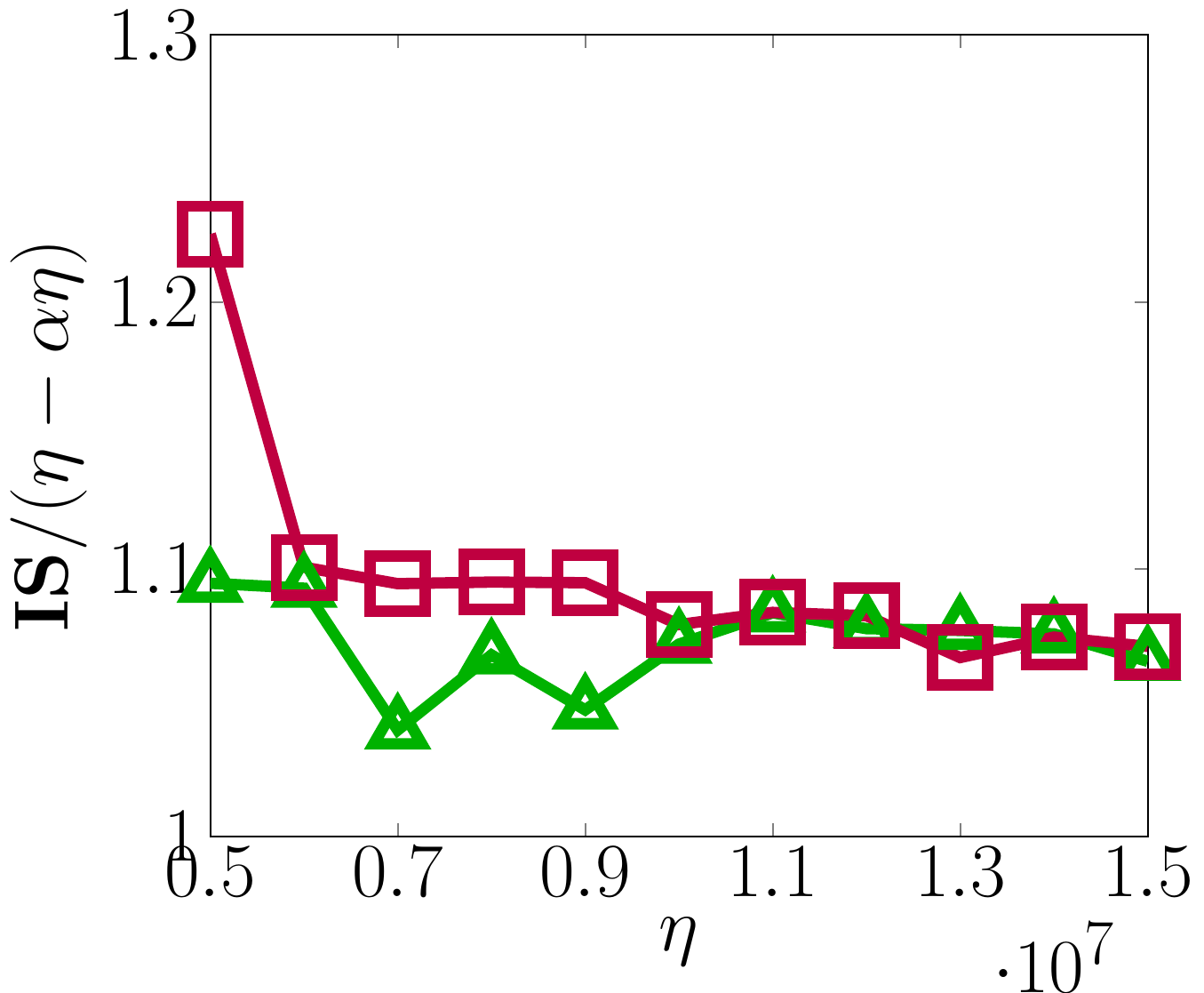}
	      \label{fig:twitter_generalCost_Influence_alpha02}
	    } \\
	
	    \subfigure[wiki-Vote (Cost)]{
	      \includegraphics[width=0.18\textwidth]{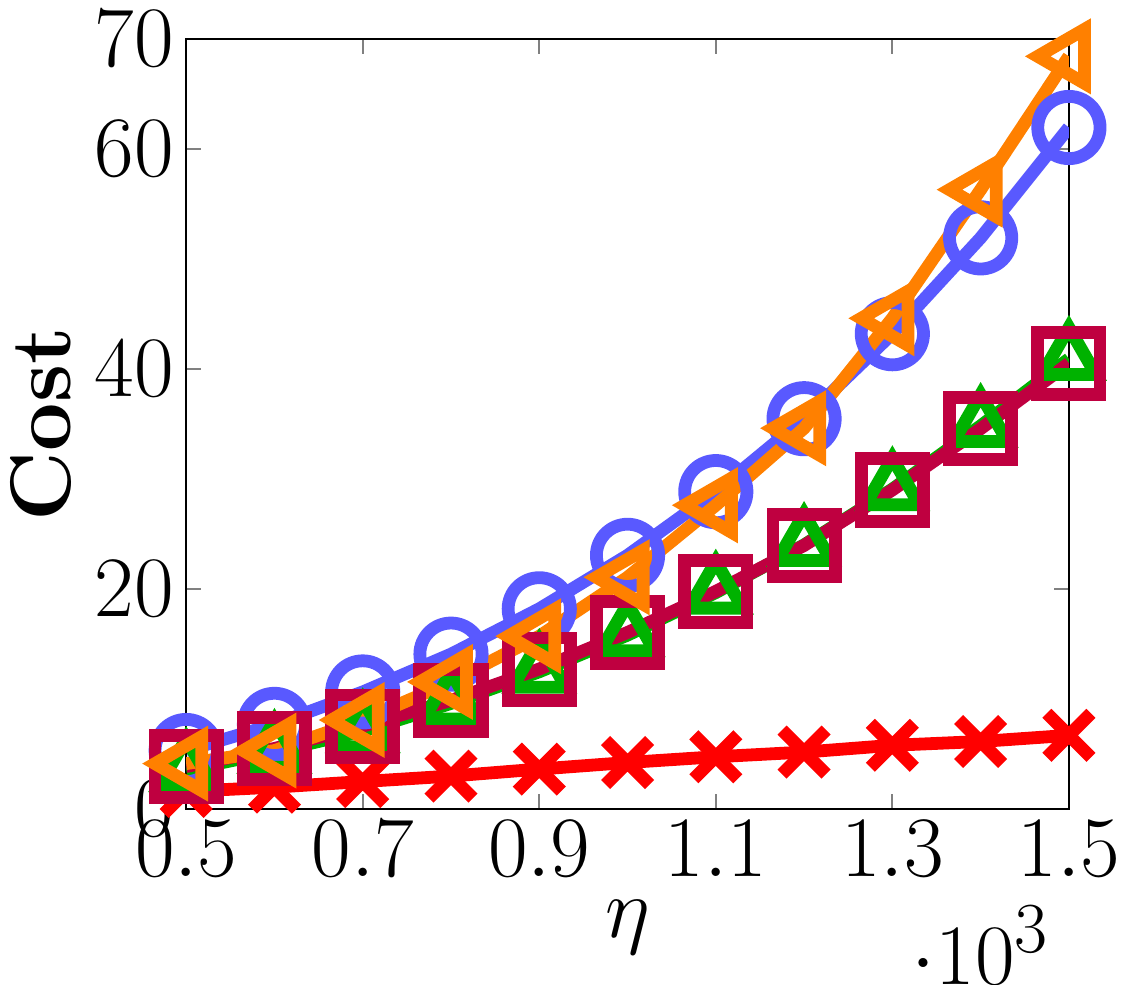}
	      \label{fig:wikiVote_generalCost_Cost_alpha02}
	    }
	    \subfigure[Pokec (Cost)]{
	      \includegraphics[width=0.18\textwidth]{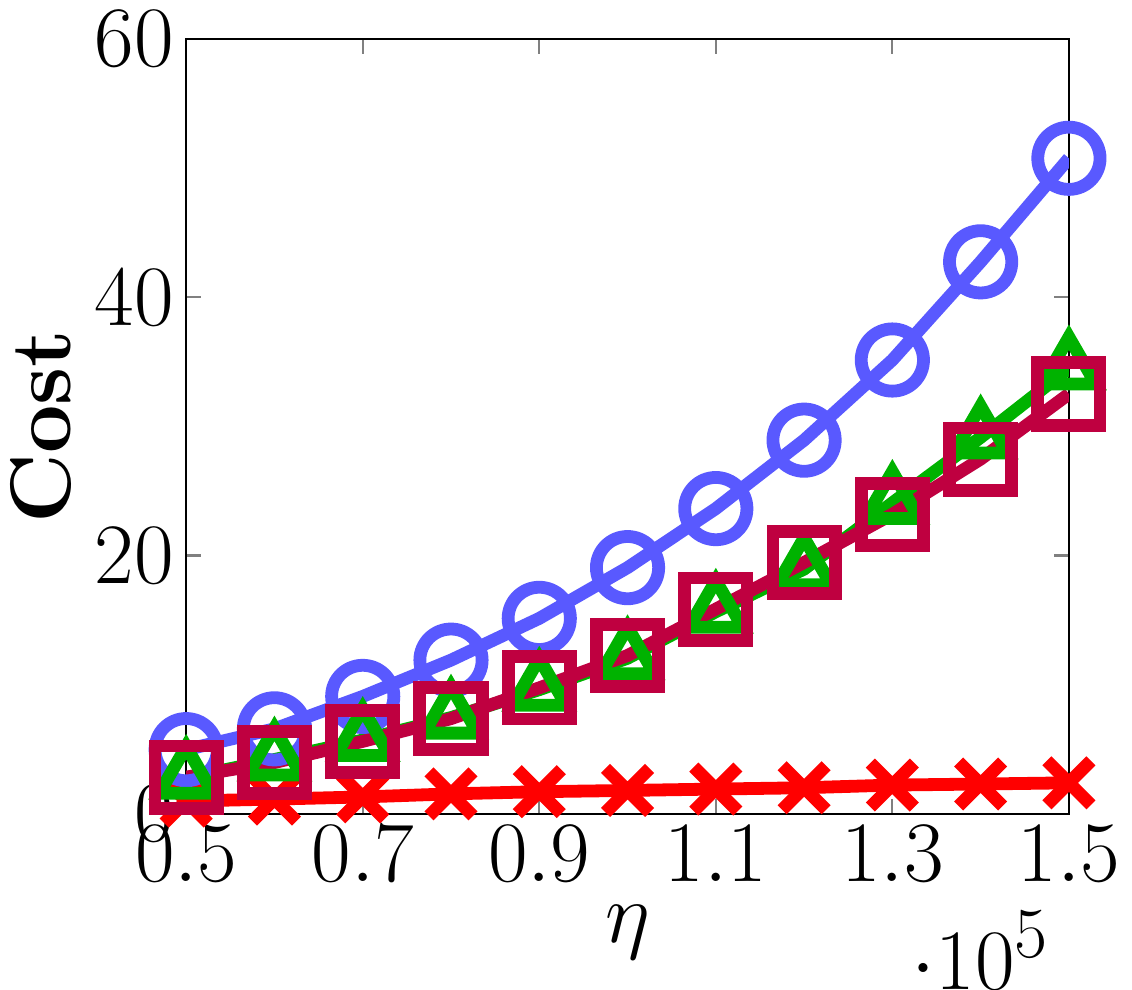}
	      \label{fig:pokec_generalCost_Cost_alpha02}
	    }
	     \subfigure[LiveJournal (Cost)]{
	      \includegraphics[width=0.18\textwidth]{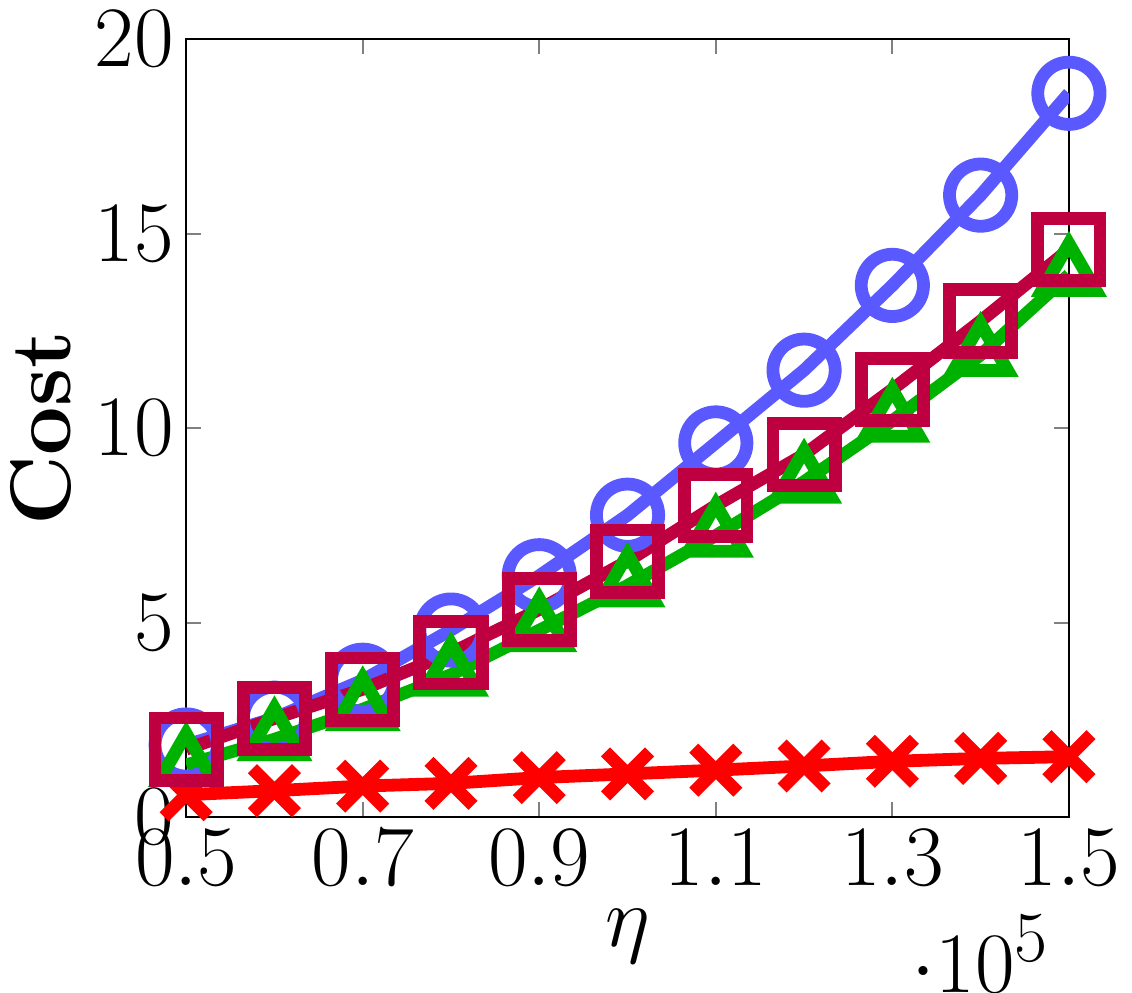}
	      \label{fig:LiveJournal_generalCost_Cost_alpha02}
	    }
	     \subfigure[Orkut (Cost)]{
	      \includegraphics[width=0.18\textwidth]{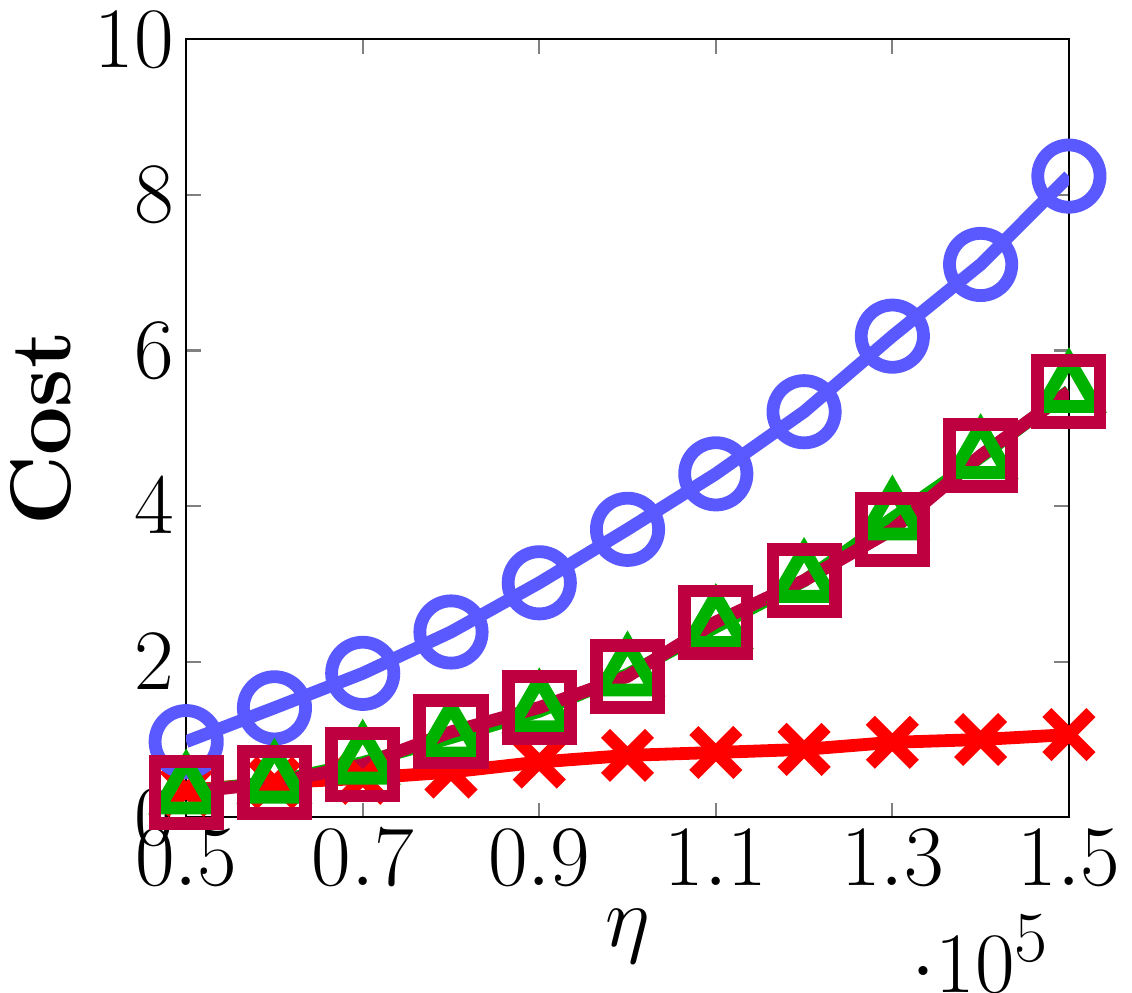}
	      \label{fig:orkut_generalCost_Cost_alpha02}
	    }
	    \subfigure[Twitter (Cost)]{
	      \includegraphics[width=0.18\textwidth]{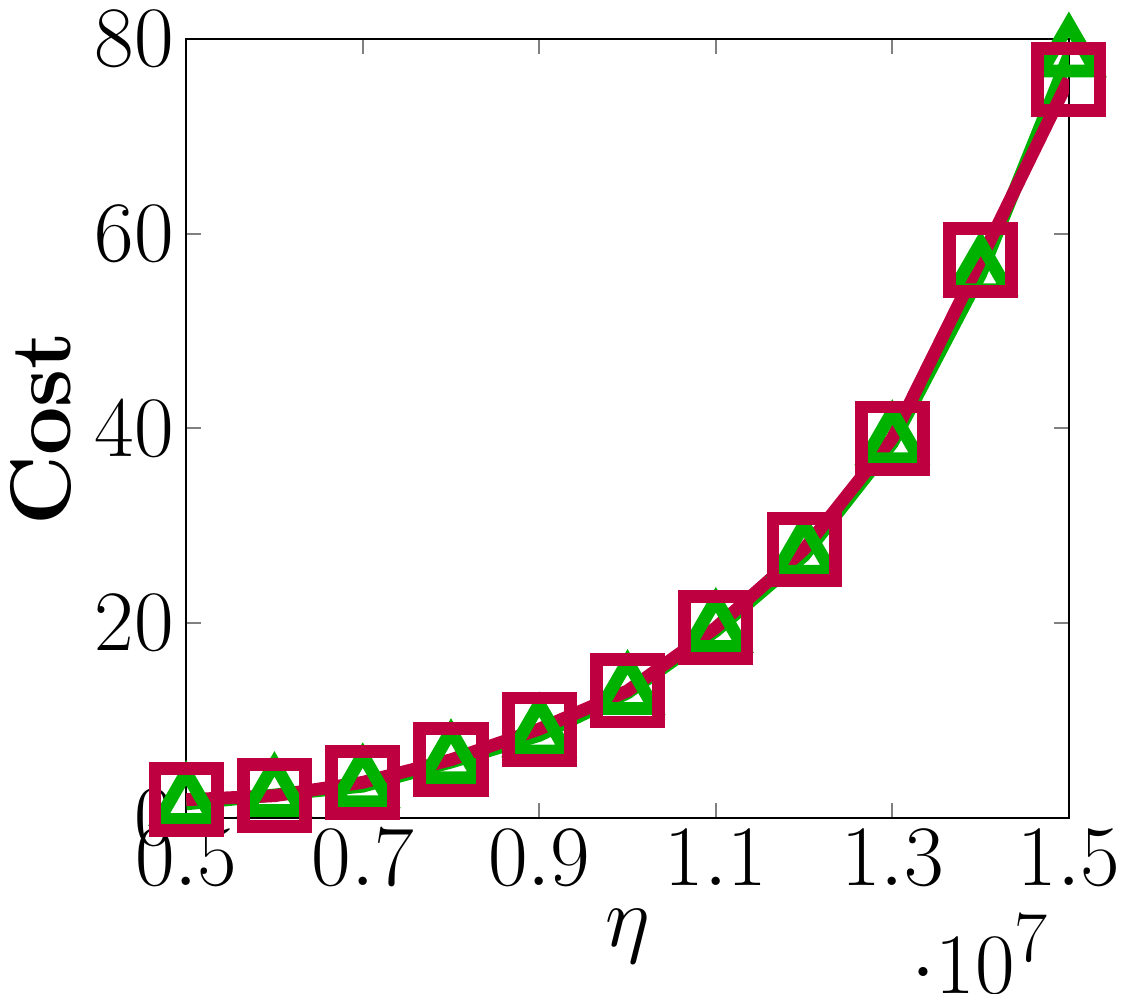}
	      \label{fig:twitter_generalCost_Cost_alpha02}
	    }\\
	
	  \end{minipage}
	   \renewcommand{\figurename}{Fig.}
	  \caption{Comparing the algorithms under the GC setting (IS: influence spread; RT: running time)}
	  \label{fig:the GC setting_alpha02}
	\end{figure*}

\begin{figure*}[htb!]
	  \centering
	   \includegraphics[width=0.65\textwidth]{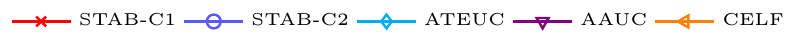}
	   \label{fig:lineTitle_ATEUC_C1_C2_AAUC_CELF_alpha02}
	
	  \begin{minipage}[htb!]{1.035\textwidth}
	  	
	    \subfigure[wiki-Vote (RT)]{
	      \includegraphics[width=0.18\textwidth]{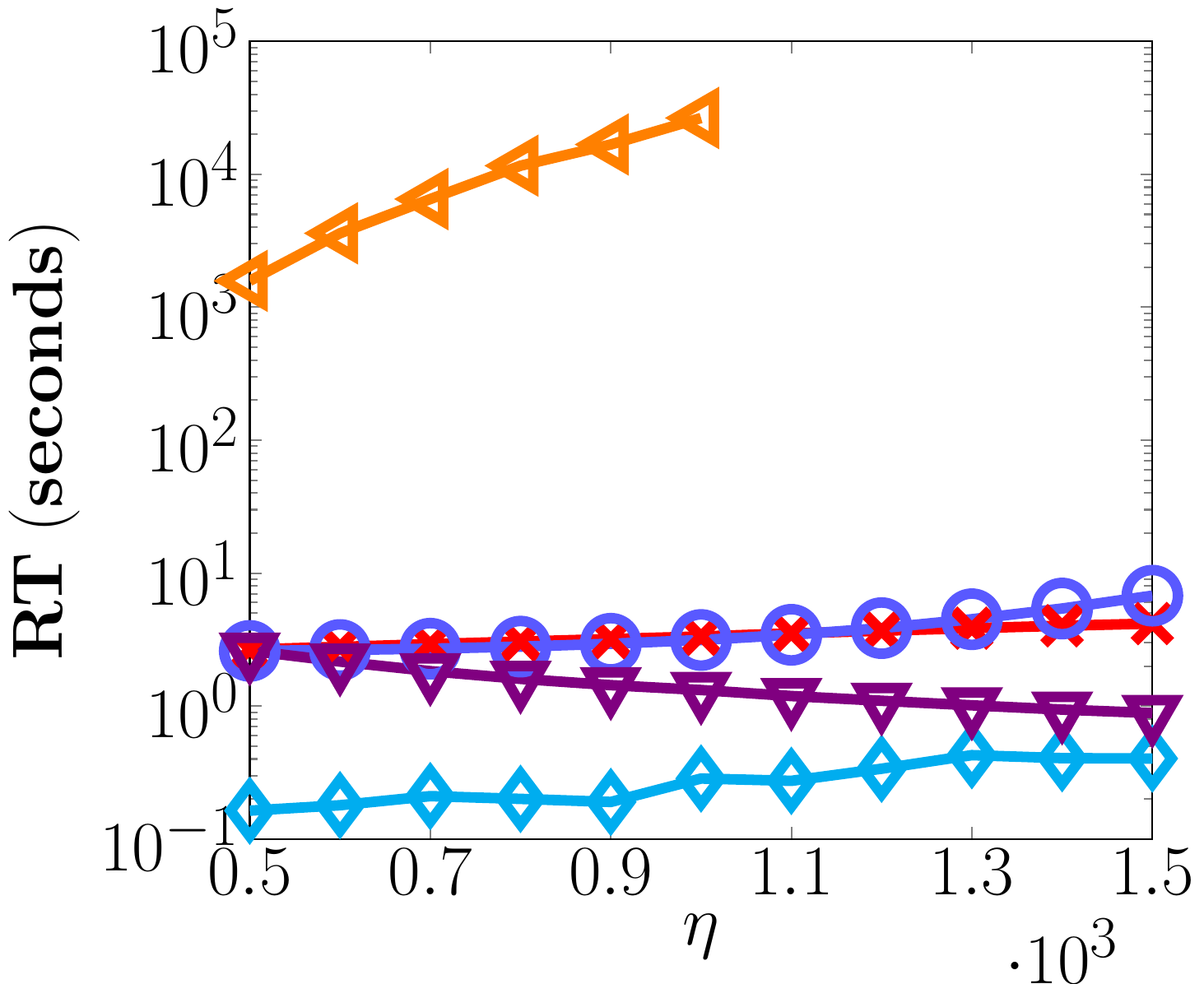}
	      \label{fig:wikiVote_uniformCost_RunningTime_alpha02}
	    }
	    \subfigure[Pokec (RT)]{
	      \includegraphics[width=0.18\textwidth]{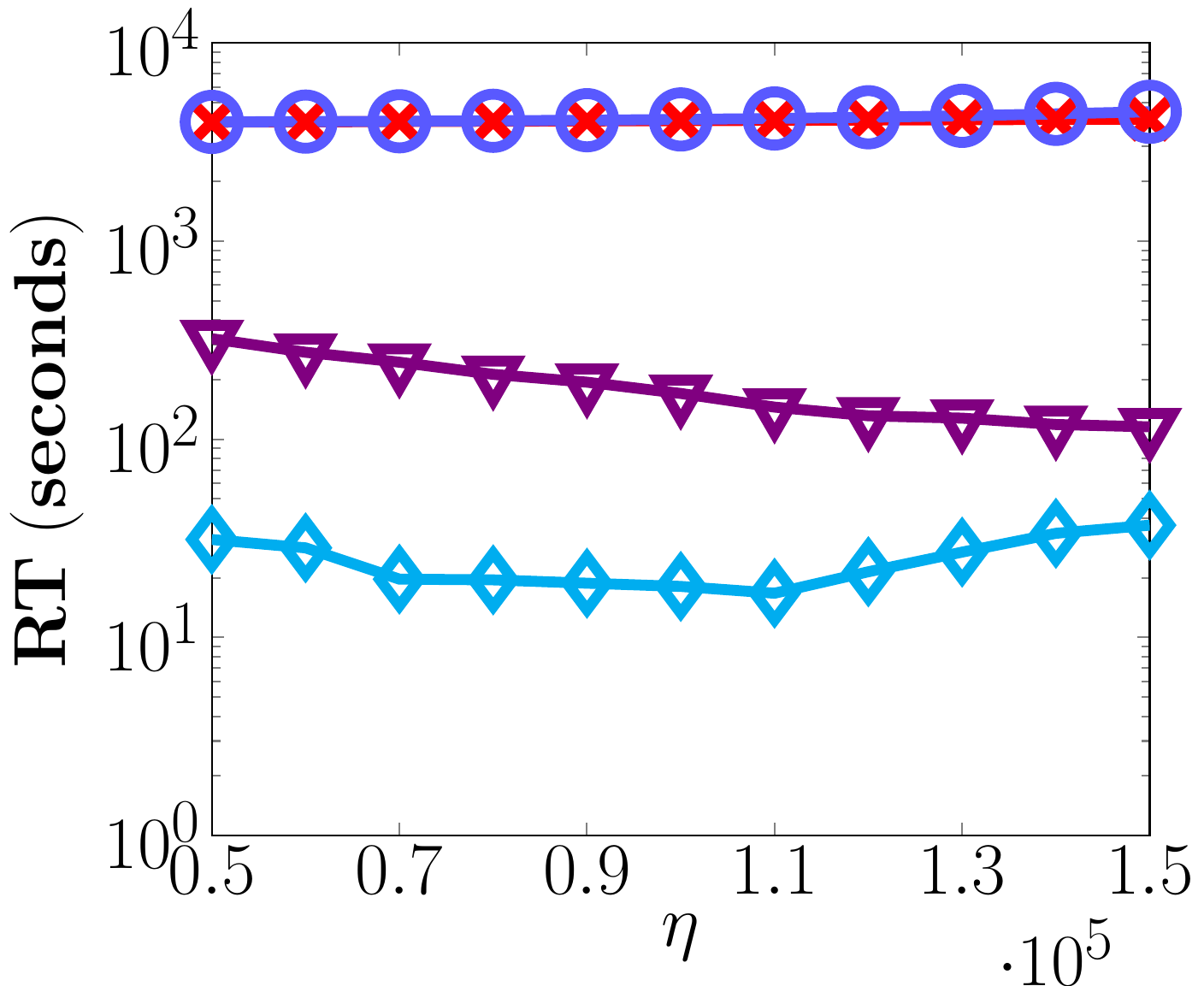}
	      \label{fig:pokec_uniformCost_RunningTime_alpha02}
	    }
	      \subfigure[LiveJournal (RT)]{
	      \includegraphics[width=0.18\textwidth]{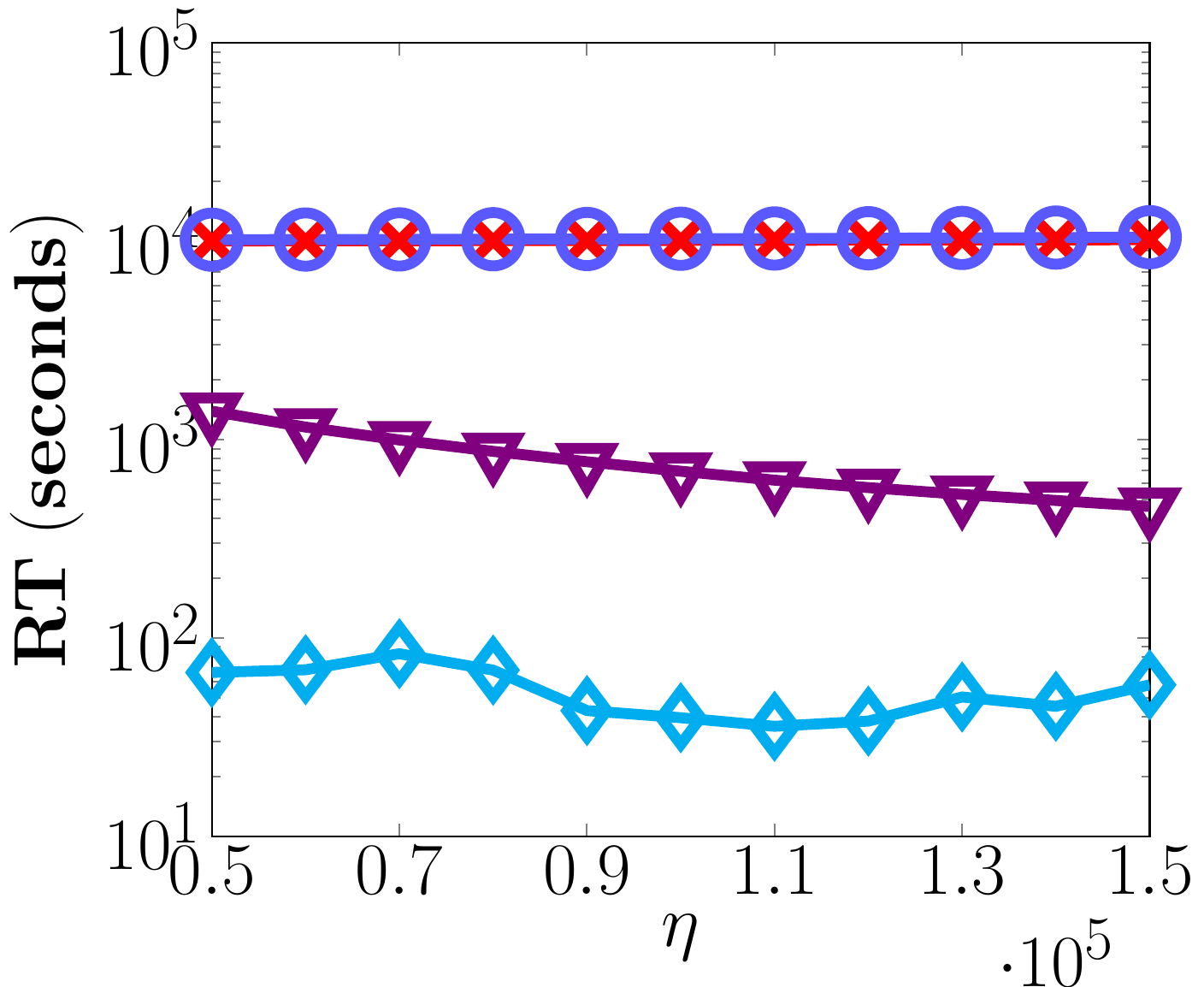}
	      \label{fig:LiveJournal_uniformCost_RunningTime_alpha02}
	    }
	     \subfigure[Orkut (RT)]{
	      \includegraphics[width=0.18\textwidth]{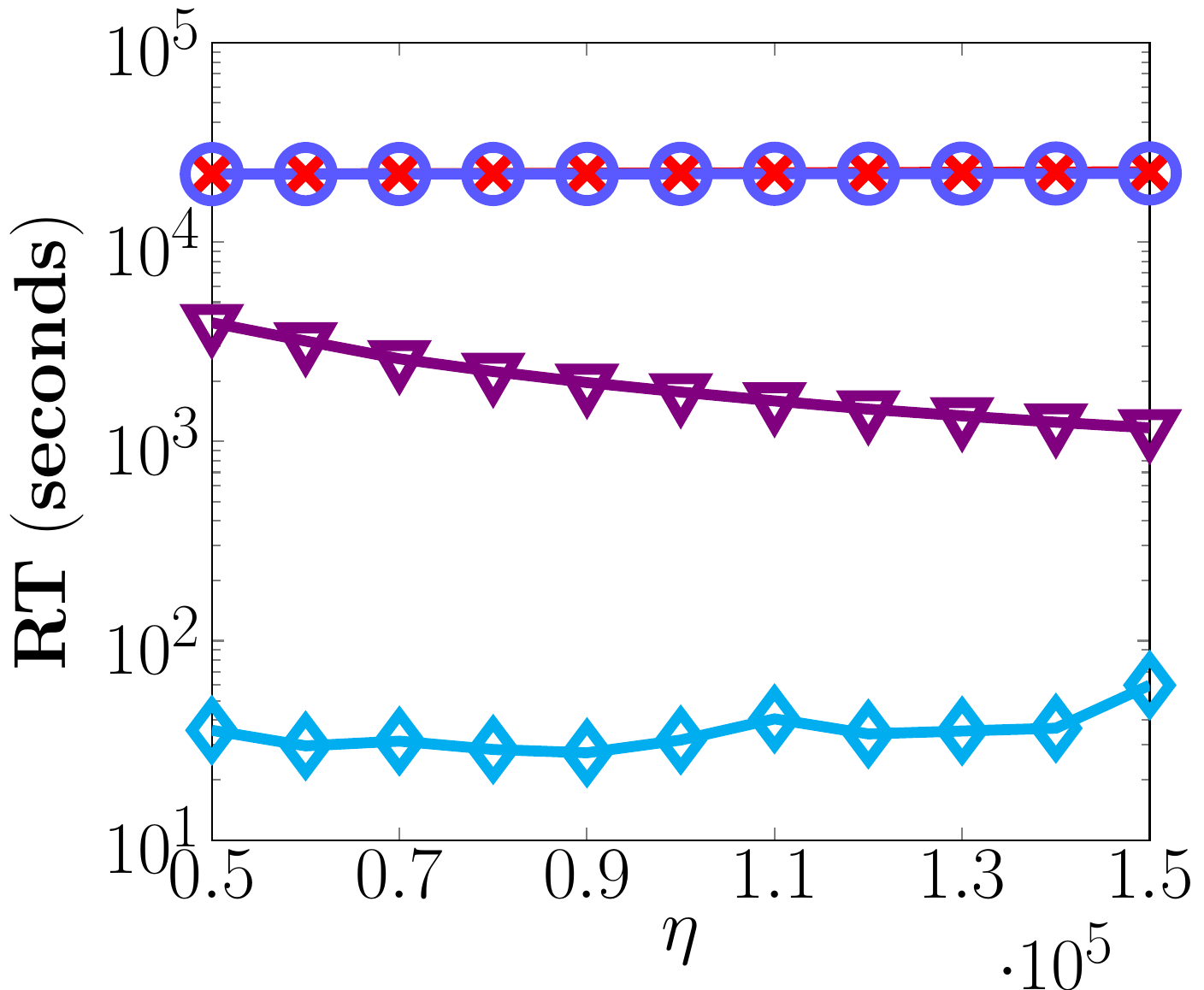}
	      \label{fig:orkut_uniformCost_RunningTime_alpha02}
	    }
	    \subfigure[Twitter (RT)]{
	      \includegraphics[width=0.18\textwidth]{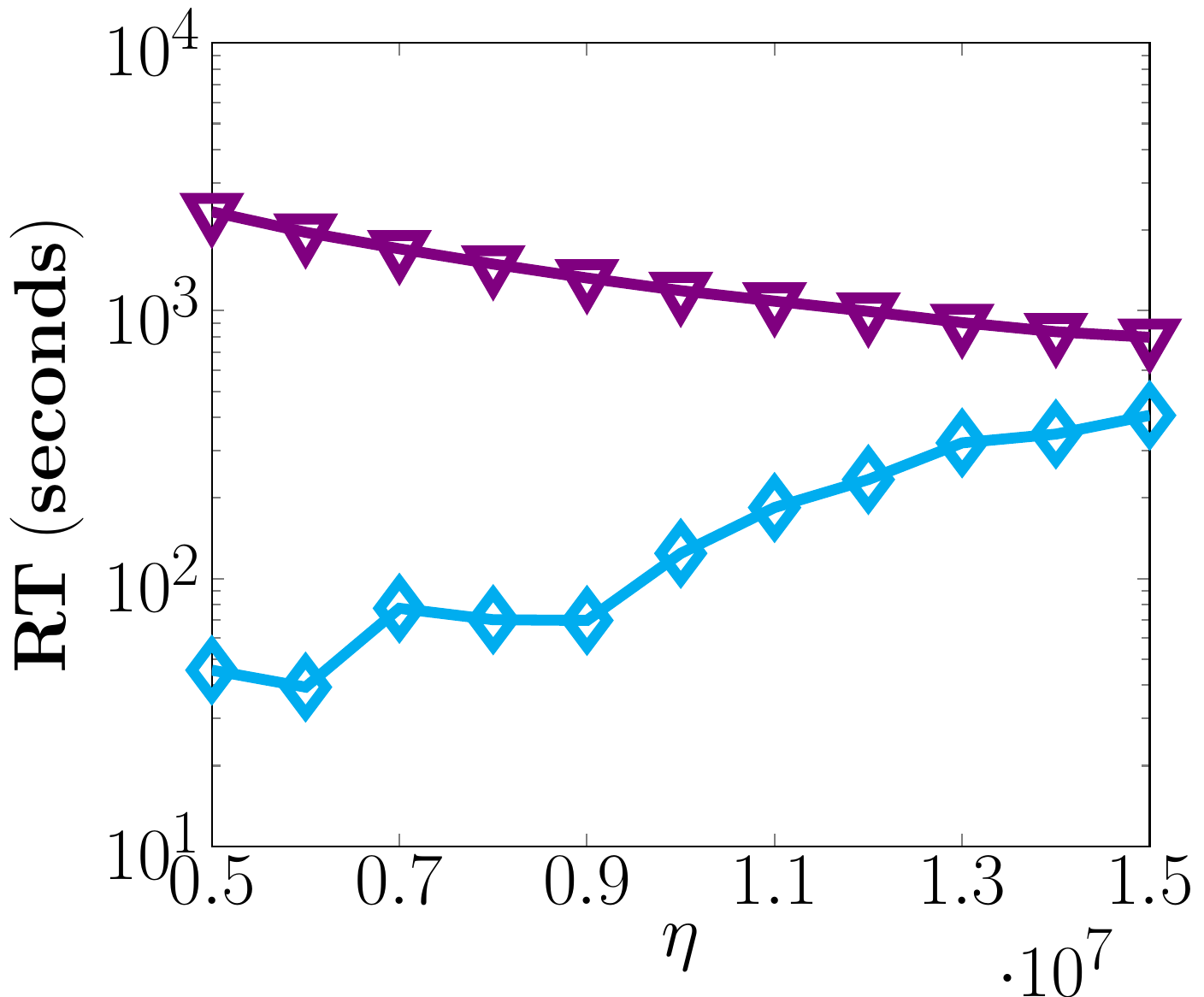}
	      \label{fig:twitter_uniformCost_RunningTime_alpha02}
	    } \\
	
	    \subfigure[wiki-Vote (IS)]{
	      \includegraphics[width=0.18\textwidth]{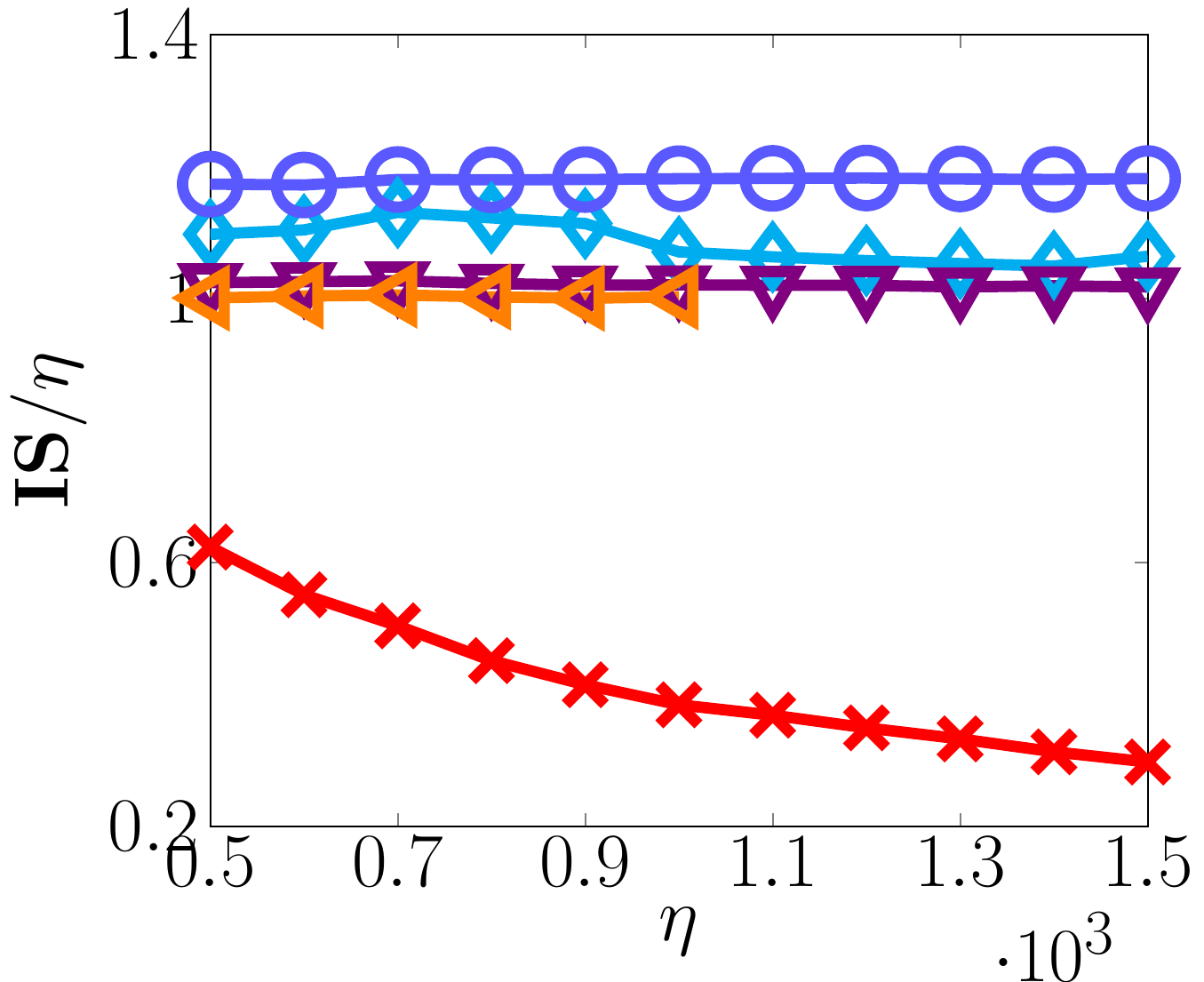}
	      \label{fig:wikiVote_uniformCost_Influence_alpha02}
	    }	
	    \subfigure[Pokec (IS)]{
	      \includegraphics[width=0.18\textwidth]{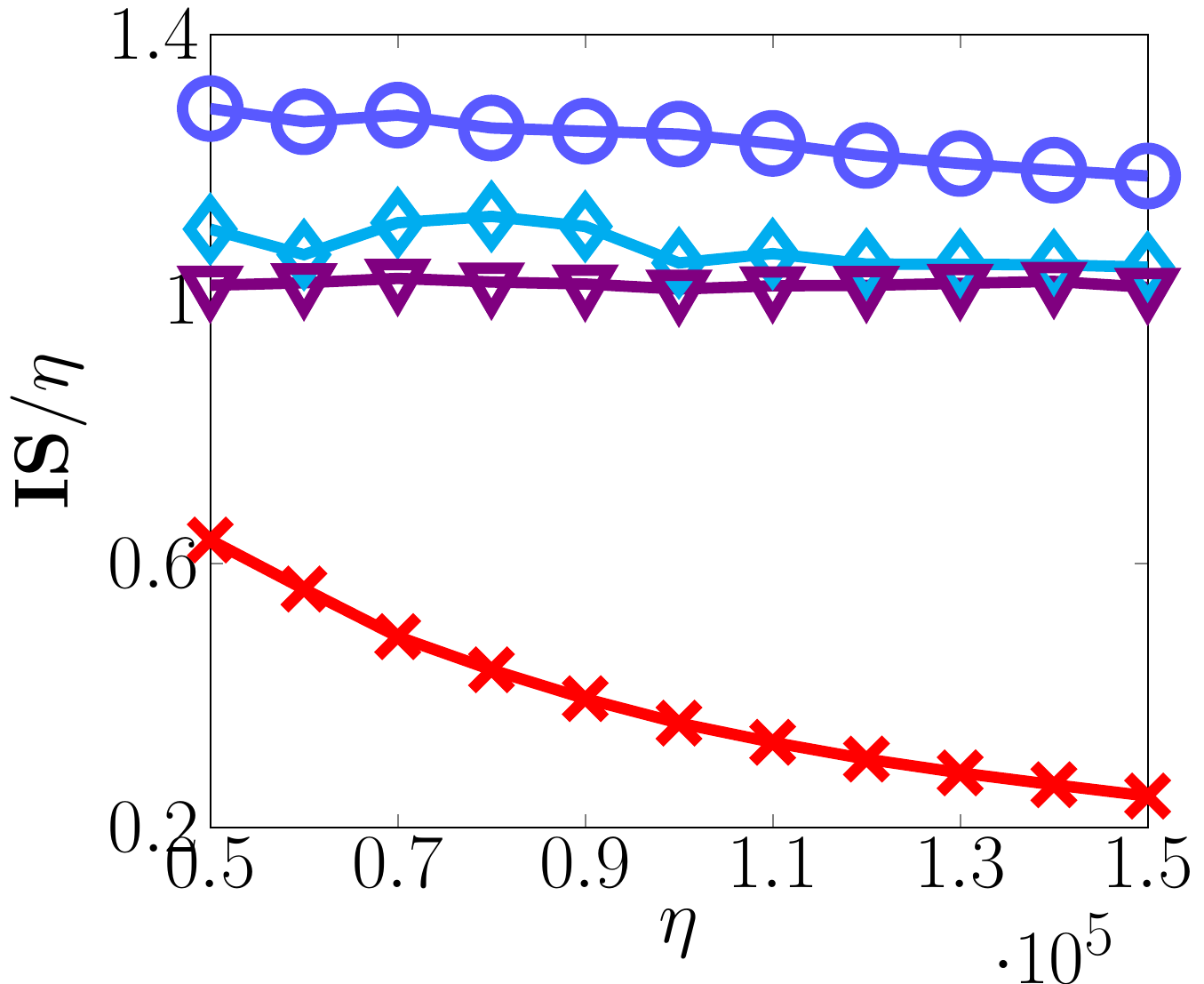}
	      \label{fig:pokec_uniformCost_Influence_alpha02}
	    }
	      \subfigure[LiveJournal (IS)]{
	      \includegraphics[width=0.18\textwidth]{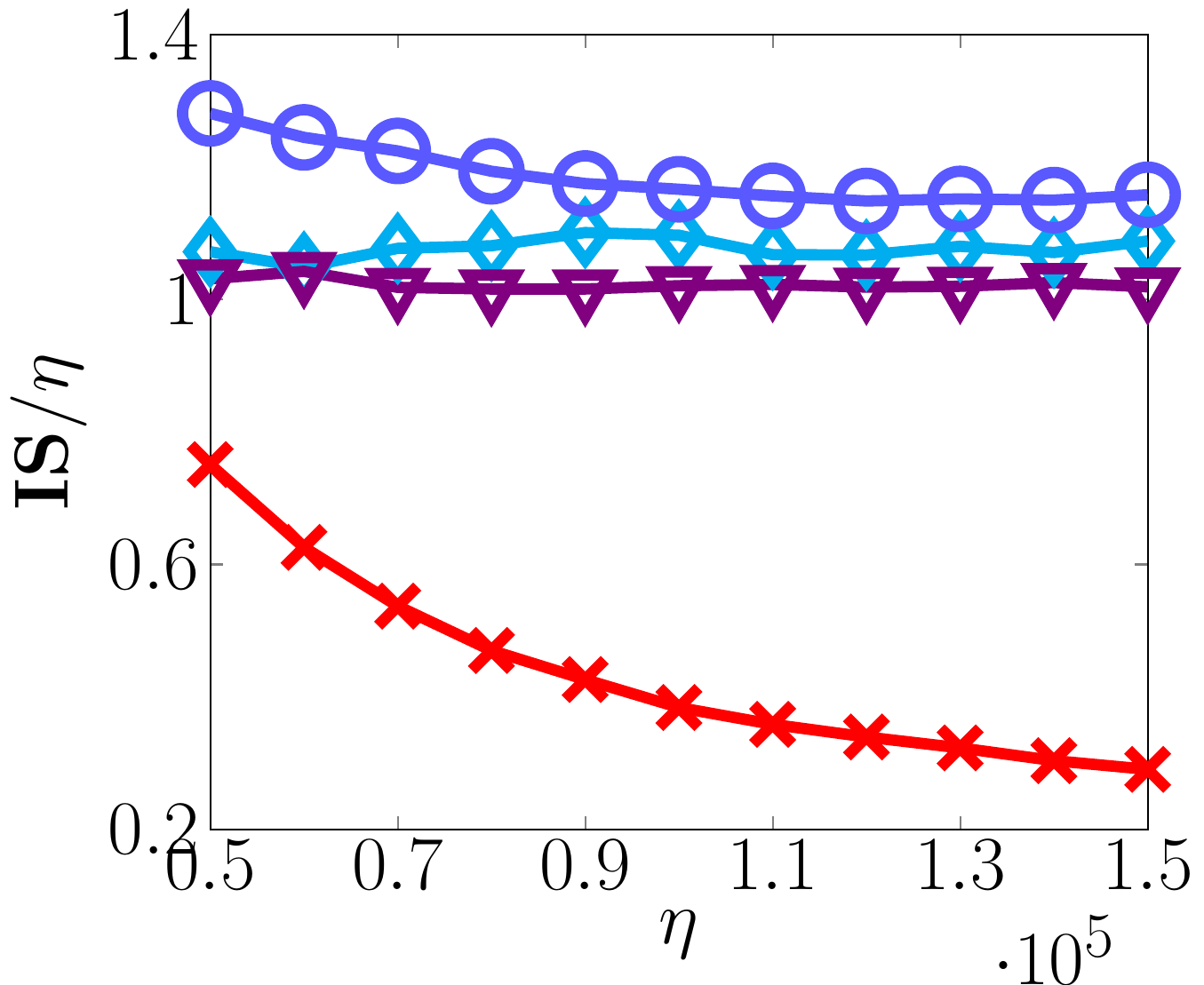}
	      \label{fig:LiveJournal_uniformCost_Influence_alpha02}
	    }
	     \subfigure[Orkut (IS)]{
	      \includegraphics[width=0.18\textwidth]{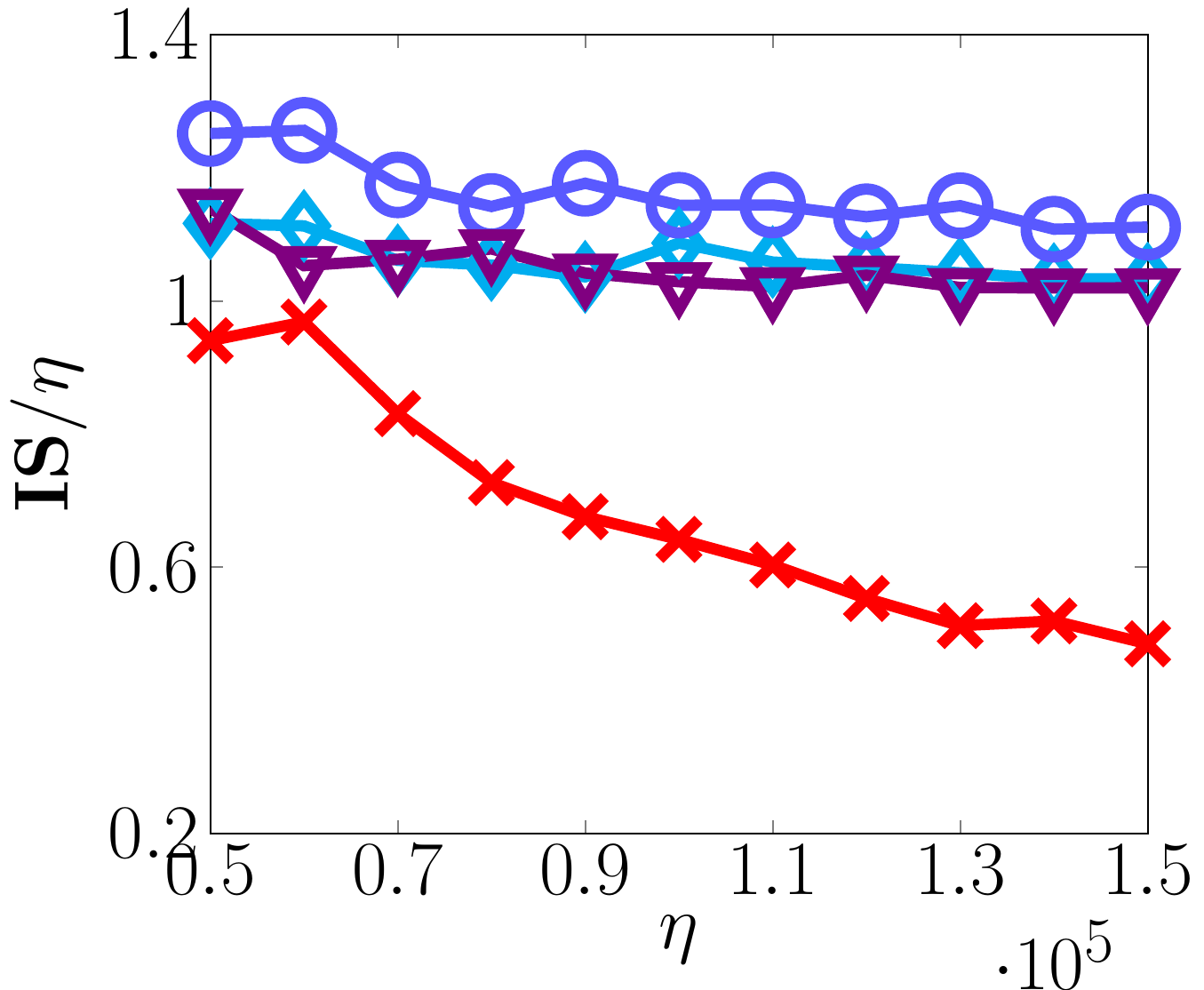}
	      \label{fig:orkut_uniformCost_Influence_alpha02}
	    }
	    \subfigure[Twitter (IS)]{
	      \includegraphics[width=0.18\textwidth]{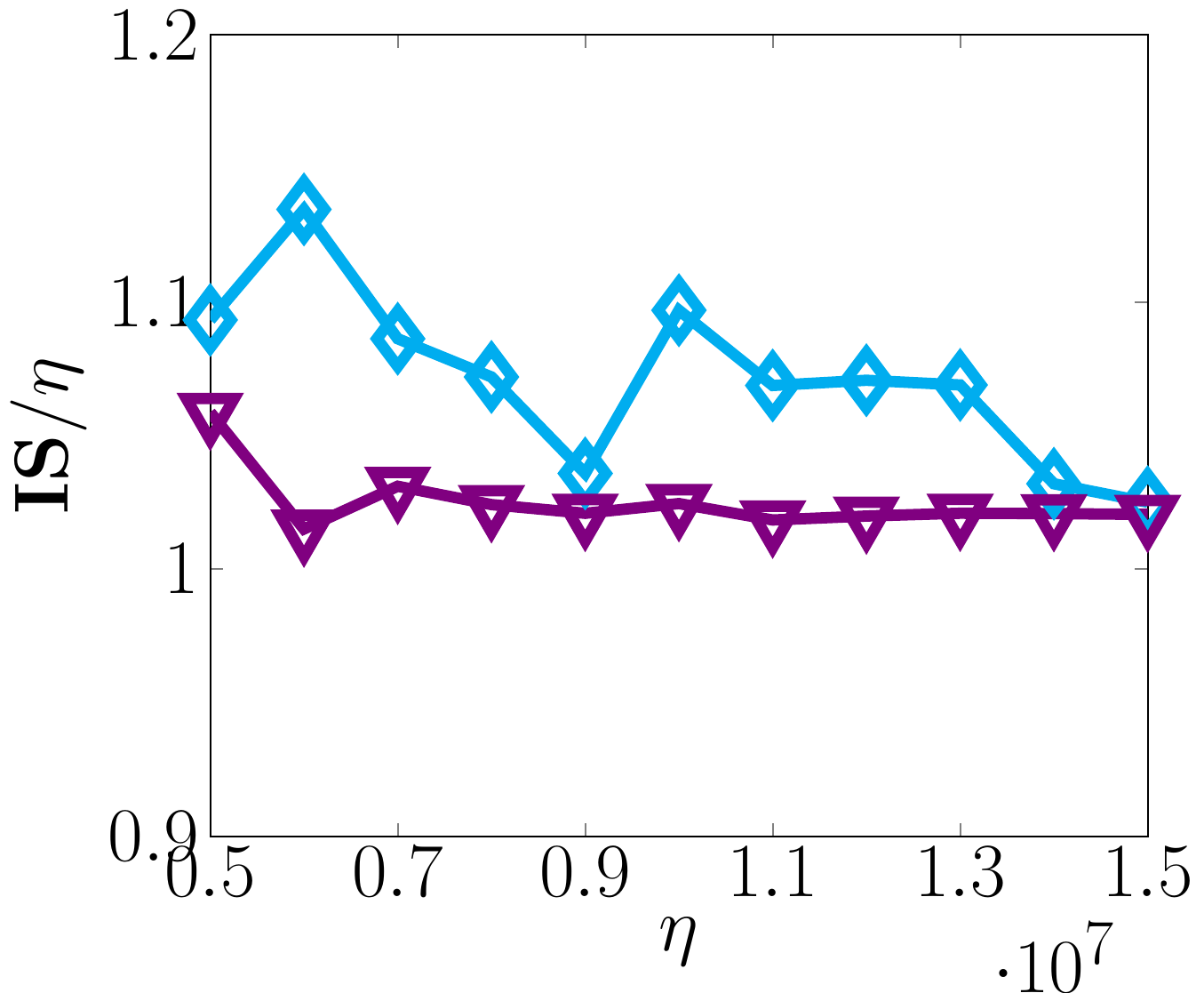}
	      \label{fig:twitter_uniformCost_Influence_alpha02}
	    } \\
	
	    \subfigure[wiki-Vote (Cost)]{
	      \includegraphics[width=0.18\textwidth]{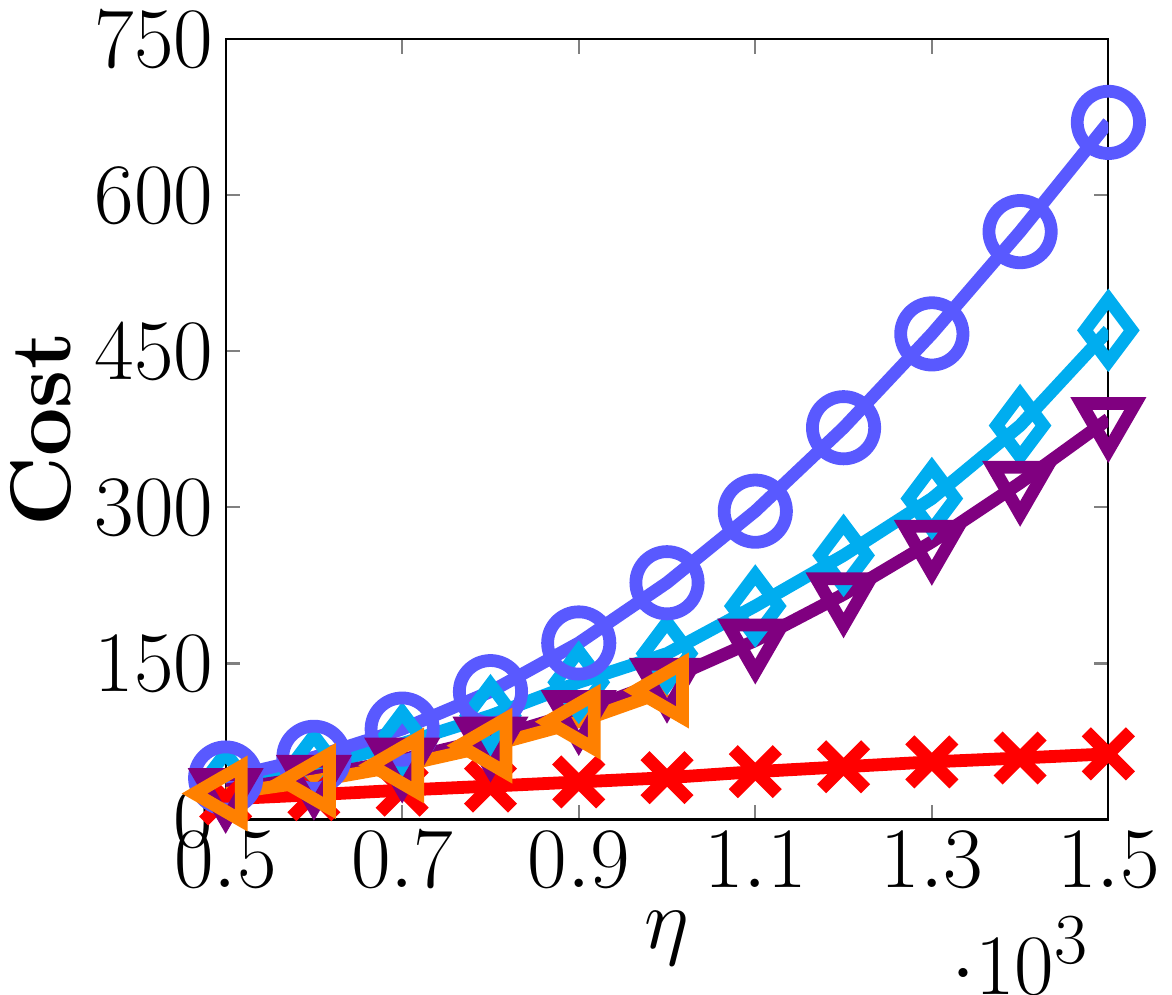}
	      \label{fig:wikiVote_uniformCost_Cost_alpha02}
	    }
	    \subfigure[Pokec (Cost)]{
	      \includegraphics[width=0.18\textwidth]{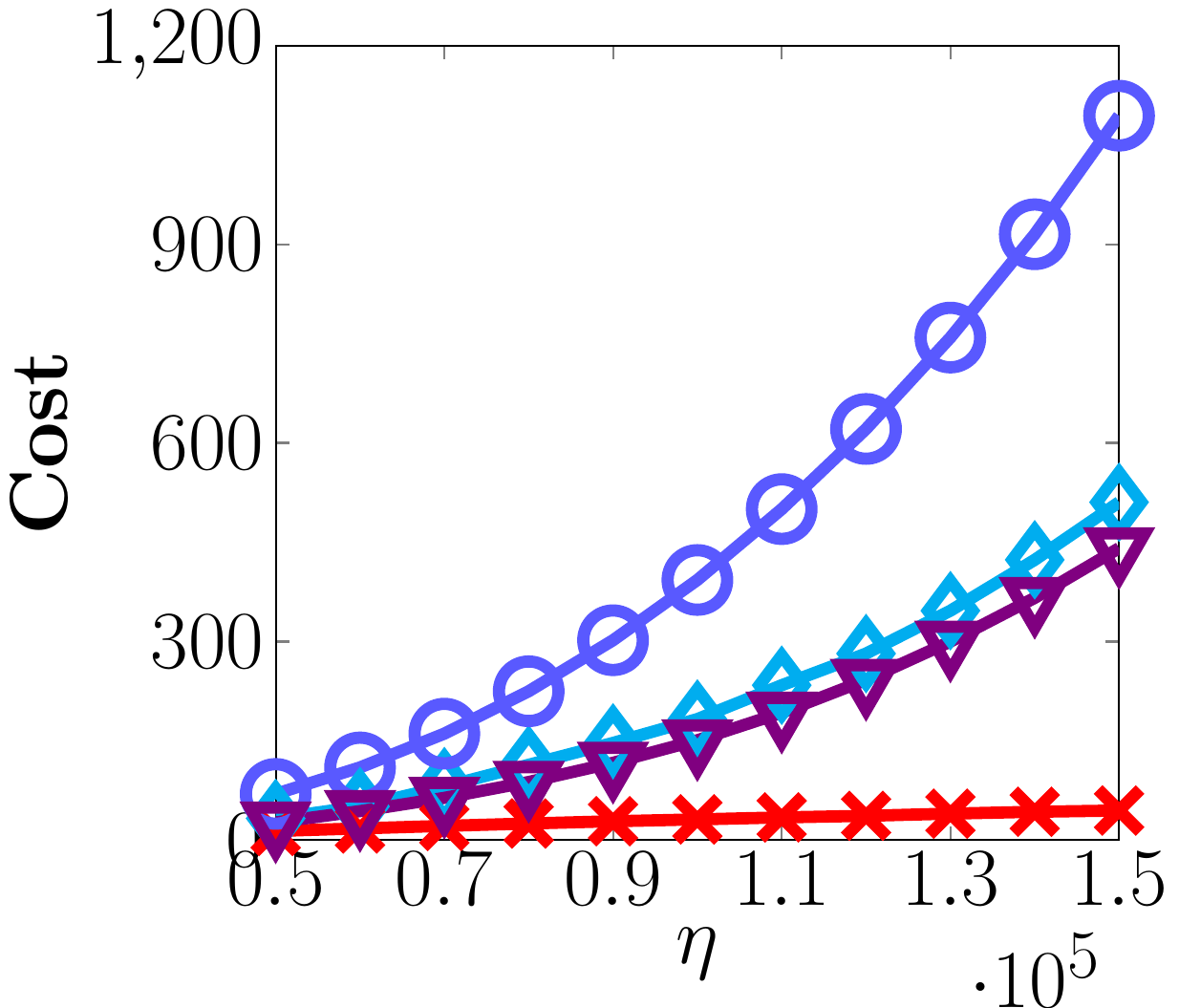}
	      \label{fig:pokec_uniformCost_Cost_alpha02}
	    }
	     \subfigure[LiveJournal (Cost)]{
	      \includegraphics[width=0.18\textwidth]{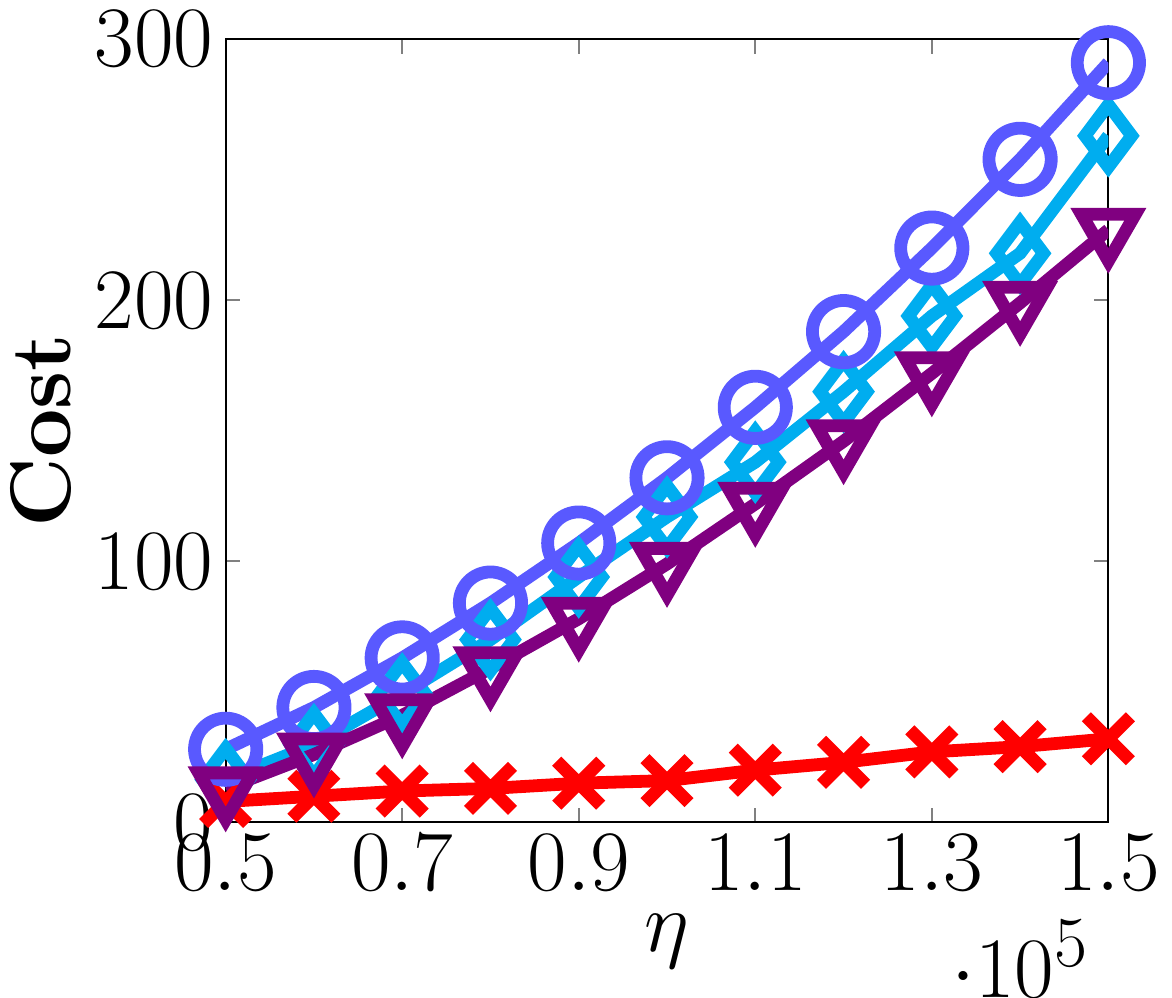}
	      \label{fig:LiveJournal_uniformCost_Cost_alpha02}
	    }
	     \subfigure[Orkut (Cost)]{
	      \includegraphics[width=0.18\textwidth]{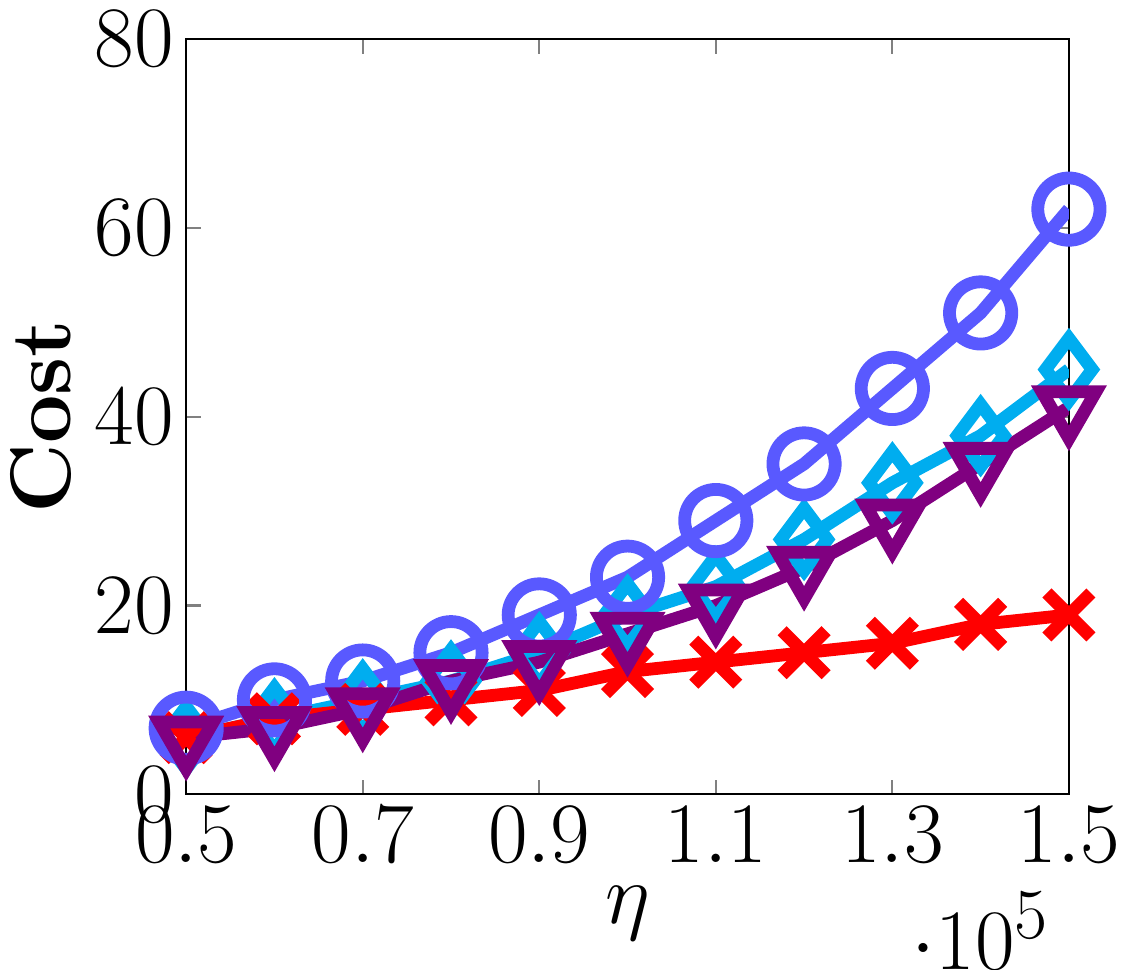}
	      \label{fig:orkut_uniformCost_Cost_alpha02}
	    }
	     \subfigure[Twitter (Cost)]{
	      \includegraphics[width=0.18\textwidth]{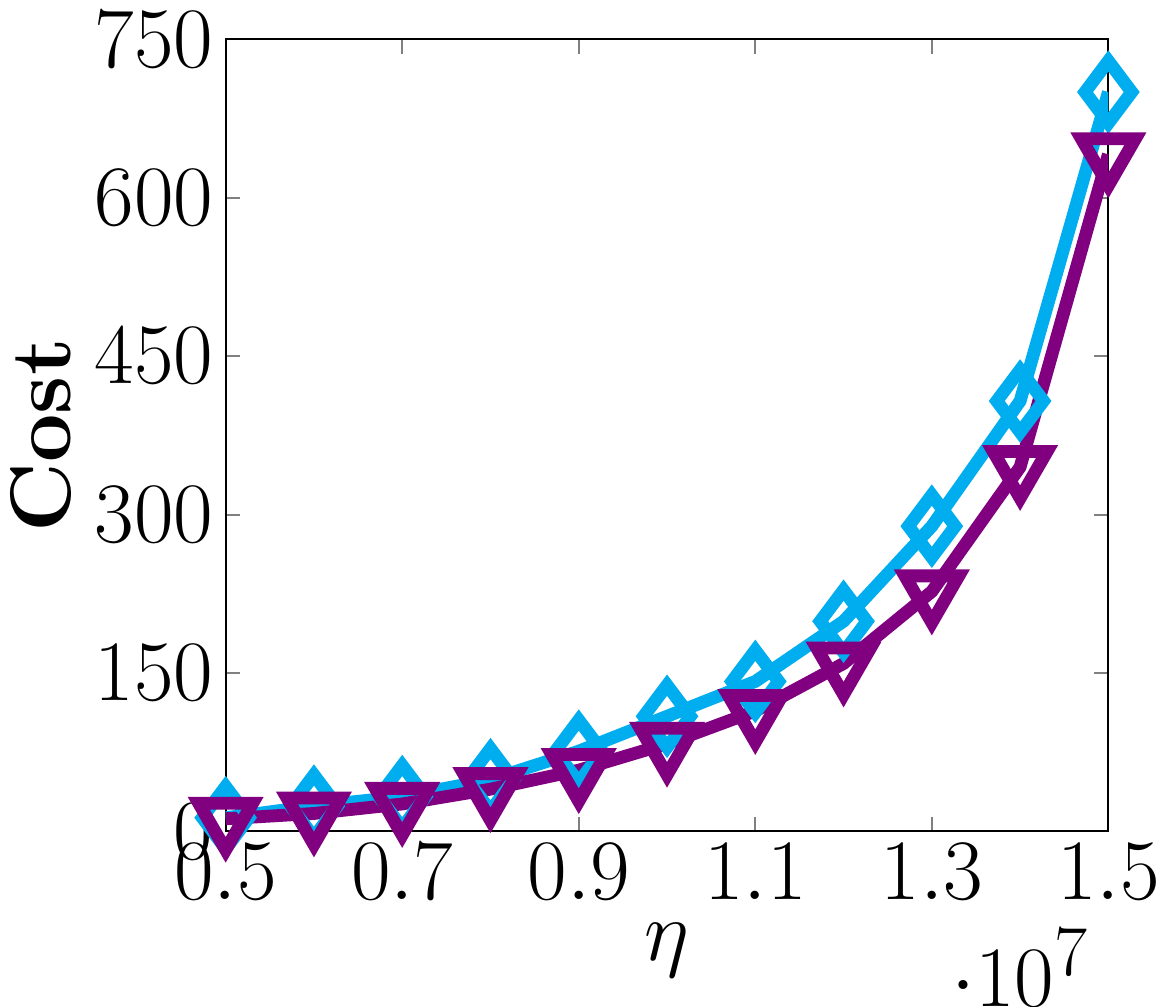}
	      \label{fig:twitter_uniformCost_Cost_alpha02}
	    } \\
	
	  \end{minipage}
	   \renewcommand{\figurename}{Fig.}
	  \caption{Comparing the algorithms under the UC setting (IS: influence spread; RT: running time)}
	  \label{fig:the UC setting_alpha02}
	\end{figure*}

\subsection{Theoretical Comparisons for the AP Algorithms}
\label{sec:compareunderuc}

To the best of our knowledge, there are no AP algorithms for MCSS with provable performance bounds under the GC+NV setting, and \textbf{only~\cite{Chen2014,Kuhnle2017} have proposed AP algorithms with provable approximation ratios under the UC+NV setting.}


\subsubsection{Comparing with Zhang et.al.'s Results~\cite{Chen2014}}

The work in~\cite{Chen2014} has proved that the $\mathsf{GreSSC}$ algorithm achieves a performance bound stated as follows:

\vspace{-1ex}
\begin{fact}[\cite{Chen2014}]
Let $S'$ be the output of $\mathsf{GreSSC}((1+\xi)\eta,\xi,\xi,0)$. Under the UC setting, if $\xi\leq \frac{\epsilon(n-\eta)}{8n^2(\eta+1)}$ where $\epsilon\in (0,1]$, then we have $f(S')\geq \eta$ and $|S'|\leq \lceil \ln\frac{(1+\epsilon)n\eta}{n-\eta} \rceil |S_{opt}|+1$ for any $\eta\in (0,n)$.
\label{fct:zhangsar}
\end{fact}


However, no previous studies including~\cite{Chen2014} have analyzed the time complexity of $\mathsf{GreSSC}$. So we analyze the time complexity of Zhang \textit{et.al.}'s algorithm by the following lemma:

\vspace{-1ex}
\begin{lemma}
Let $l_i=\min\{f(A)| |A|=i\}$ for any $i\in \{1,\cdots,n\}$ and $\delta$ be any number in $(0,1)$. With probability of at least $1-\delta$, Zhang \textit{et.al.}'s algorithm~\cite{Chen2014} can find a solution $S'$ satisfying the performance bound shown in Fact~\ref{fct:zhangsar} under the time complexity of $\mathcal{O}(\sum_{i=1}^{|S'|} \frac{n^2 m}{l_i\xi^2}\ln\frac{|S'|}{\delta})$.
\label{lma:greedystimecomplexity}
\end{lemma}

Recall that $\mathsf{AAUC}$ outputs $S$ satisfying $f(S)\geq \eta$ and $|S|\leq \lceil \ln\frac{n\eta}{n-\eta}\rceil |{S}_{opt}|+2$ w.h.p. under $\mathcal{O}(\frac{m q}{\varrho^2}\ln\frac{n}{\delta\eta})$ running time, where $q=\max\{f(v)|v\in V\}$ (see Theorem~\ref{thm:timeofaauc}). According to Lemma~\ref{lma:greedystimecomplexity}, even if we set $|S'|=1$, $l_i=n (\forall i)$ and $q=n$ in favor of Zhang \textit{et.al.}'s algorithm, the time complexity of $\mathsf{AAUC}$ is still $\Omega(n^2/\ln n)$ times smaller than that of Zhang \textit{et.al.}'s algorithm, as $\varrho=\Omega(n \xi)$.

Moreover, Theorem~\ref{thm:arofaauc} shows that, although $\mathsf{AAUC}$'s approximation ratio is larger than that in Fact~\ref{fct:zhangsar} by an additive factor 1, its multiplicative factor (i.e., $\lceil \ln\frac{n\eta}{n-\eta}\rceil$) is always smaller than that in Fact.~\ref{fct:zhangsar} (i.e., $\lceil \ln\frac{(1+\epsilon)n\eta}{n-\eta} \rceil$) due to $\epsilon>0$ (note that smaller $\epsilon$ results in higher time complexity and $\epsilon=0$ implies the impractical EV setting). As $\mathsf{AAUC}$ is $\Omega(n^2/\ln n)$ times faster than Zhang \textit{et.al.}'s algorithm even when we set $\epsilon=1$ in Fact~\ref{fct:zhangsar}, the approximation ratio of $\mathsf{AAUC}$ can be better than Zhang \textit{et.al.}'s algorithm~\cite{Chen2014} under the same running time. 

\subsubsection{Comparing with Kuhnle et.al.'s Results~\cite{Kuhnle2017}}
The work in~\cite{Kuhnle2017} has also proposed a theoretical bound stated as follows:
\begin{fact}
For any $u\in V$, let $r_u$ be the probability that $u$ remains inactive when all of the neighbors of $u$ in G are activated. Let $r^*=\min_{u\in V} r_u$. There exists an algorithm that can output a set $S'$ satisfying $f(S')\geq\eta$ w.h.p. and $|S'|\leq \left( \frac{1+\varepsilon}{r^*}+\log\frac{\eta}{|S_{opt}|}\right) |S_{opt}|$.
\label{fct:mythaiaa}
\end{fact}

Note that the parameter $r^*$ in Fact.~\ref{fct:mythaiaa} can be very small. For example, in the IC model, we have $r^*=\min_{u\in V} \Pi_{v\in N_{in}(u)}(1-p(v,u))$, where $N_{in}(u)=\{v\in V: \langle v,u\rangle \in E\}$ and $p(v,u)$ is probability that $v$ can activate $u$. Whenever there exists an edge $\langle v,u\rangle\in E$ with $p(v,u)=1$, we will have $r^*=0$ and hence the approximation ratio in Fact~\ref{fct:mythaiaa} becomes infinite. In contrast, $\mathsf{AAUC}$ and $\mathsf{ATEUC}$ always have logarithmic approximation ratios which do not depend on the influence propagation probabilities of the network.


The work in~\cite{Kuhnle2017} has not explained how to implement an algorithm that can achieve the theoretical bounds shown in Fact.~\ref{fct:mythaiaa}, so its time complexity is unclear (note that we cannot follow Fact~\ref{fct:kuhapproxalg} to implement an AP algorithm by setting $\rho=0$, because this would result in infinite time complexity).

\section{Performance Evaluation}
\label{sec:pe}
In this section, we evaluate the performance of our algorithms through extensive experiments on real OSNs.



\subsection{Experimental Settings}
\label{sec:expsetting}

We use public OSN datasets in the experiments, which are shown in Table~\ref{table:datasets}. These data sets are widely used in the literature~\cite{TangSX2015,TangXS2014,NguyenTD2016}, and they can be downloaded from~\cite{snap,TangSX2015}.



\begin{table}[h]
   \centering
   \small
   \begin{tabular}{|c|r|r|c|c|}\hline
	    Name & \makecell[c]{$n$} & \makecell[c]{$m$} & Type & Average degree\\ \hline
	    wiki-Vote & 7.1K & 103.7K & directed & 29.1\\ \hline
	    Pokec & 1.6M & 30.6M & directed & 37.5\\ \hline
	    LiveJournal & 4.8M & 69.0M & directed & 28.5\\ \hline
	    Orkut & 3.1M & 117.2M & undirected & 76.3\\ \hline
	    Twitter & 41.7M & 1.5G & directed & 70.5\\ \hline
    \end{tabular}
    \setcounter{table}{0}
    \caption{Datasets}
    \label{table:datasets}
\end{table}

Our experiments are conducted on a linux PC with an Intel(R) Core(TM) i7-6700K 4.0GHz CPU and 64GB memory for the Twitter dataset, while we reduce the memory to 32GB for the other datasets. Following~\cite{Kuhnle2017}, we use the independent cascade model~\cite{Kempe2003} in the experiments. For any $\langle u,v\rangle\in E$, the probability that $u$ can activate $v$ is set to $1/d_{in}(v)$, where $d_{in}(v)$ is the in-degree of $v$. Under the GC setting, the cost of each node is randomly sampled from the uniform distribution with the support $(0,1]$. These settings have been widely adopted in the literature~\cite{TangSX2015,TangXS2014,NguyenTD2016,NguyenZ2013,Goyal2013}.
As in~\cite{TangXS2014,TangSX2015,NguyenTD2016}, our reported data are the average of 10 runs, and we also limit the running time of each algorithm in a run to be within 500 minutes.
The implemented algorithms include:

\subsubsection{CELF} \label{sec:celf}
Following~\cite{Kuhnle2017,Goyal2013}, we implement $\mathsf{GreSSC}$ by adapting the the CELF algorithm~\cite{LeskovecKGFVG2007}, which has used a lazy evaluation technique to reduce the time complexity. CELF uses 10000 monte-carlo simulations to estimate $f(A)$ for any $A\subseteq V$ whenever needed~\cite{LeskovecKGFVG2007}. We notice that Zhang \textit{et.al.}~\cite{Chen2014} have also implemented $\mathsf{GreSSC}$ for Fact~\ref{fct:zhangsar}, but they have not used the parameters in Fact~\ref{fct:zhangsar} and simply used 10000 monte-carlo simulations to estimate $f(A) (\forall A\subseteq V)$. 
Therefore, the implementation of $\mathsf{GreSSC}$ in~\cite{Chen2014} does not guarantee the performance bound shown in Fact~\ref{fct:zhangsar} and is essentially the same with CELF. 

\subsubsection{STAB-C1 and STAB-C2} These algorithms (from~\cite{Kuhnle2017}) are the state-of-the-art algorithms for the MCSS problem, and they are based on the min-hash sketches proposed by~\cite{CohenDPW2014}. Moreover, it is shown in~\cite{Kuhnle2017} that these algorithms significantly outperform the traditional $k$SS algorithms. However, both STAB-C1 and STAB-C2 are BA algorithms designed only for the UC setting, so we have to adapt them to our case. As they are both greedy algorithms, we adapt them to the GC setting by selecting the node $u^*=\arg\max_{u}\Delta_u/C(\{u\})$ at each step, where $\Delta_u$ is the marginal gain computed by their estimators. To convert them into AP algorithms, we change their stopping condition to $\hat{f}(A)\geq (1+\rho)\eta$ ($\hat{f}(A)$ is their estimated IS of any $A\subseteq V$), because otherwise the solutions output by them are found to be infeasible in our experiments. We also follow the other parameter settings in~\cite{Kuhnle2017} and set $\iota=0.01,\rho=0.2$~\cite{Kuhnle2017}, where $\iota$ and $\rho$ are clarified in Fact~\ref{fct:kuhapproxalg}. 

\subsubsection{Our Algorithms} We implement $\mathsf{AAUC}$, $\mathsf{BCGC}$, $\mathsf{TEGC}$ and $\mathsf{ATEUC}$ in our experiments, where we set $\delta=0.01$ in favor of the STAB algorithms\footnote{Note that $\delta$ corresponds to $\iota n^3$ in Fact.~\ref{fct:kuhapproxalg}. However, if we set $\delta=\iota n^3$, then $\iota$ would be very small, and hence the time complexity of the STAB algorithms would be much larger than that in the current setting.}. Note that none of the algorithms compared to us has provable performance bounds in our case. More specifically, the STAB algorithms only provide bi-criteria performance bounds under the UC case; and Zhang~\textit{et.al.}'s implementation~\cite{Chen2014} is essentially the same as CELF, which does not obey the approximation ratio stated in Fact~\ref{fct:zhangsar} but has much less time complexity (see \ref{sec:celf}). Therefore, for fair comparison, we set $\sigma=\gamma={\alpha}/{3}$, $\tau=0.02$ and $\mu=n^8$ to implement our algorithms. Under these parameter settings, the $\mathsf{BCGC}$ and $\mathsf{AAUC}$ algorithms may not obey the theoretical performance bounds shown in Theorem~\ref{thm:approximationratio} and Theorem~\ref{thm:arofaauc}, as these bounds require $\mu={\Theta}(D(\eta))$. However, both $\mathsf{TEGC}$ and $\mathsf{ATEUC}$ obey the theoretical performance bounds shown in Theorem~\ref{thm:artegc} and Theorem~\ref{thm:arofateuc} in all our experiments, and they scale well to billion-scale networks (we will see this shortly).
Finally, we set $\alpha=0.2$ in our algorithms, which corresponds to the setting of $\rho=0.2$ in the STAB algorithms. 


\subsection{Experimental Results}
In Fig.~\ref{fig:the GC setting_alpha02}, we compare the algorithms under the GC case. 
It can be seen from Fig.~\ref{fig:wikiVote_generalCost_RunningTime_alpha02}-\ref{fig:twitter_generalCost_RunningTime_alpha02} that CELF runs most slowly, and its running time exceeds the time limit for all the datasets except wiki-Vote.
Moreover, both $\mathsf{BCGC}$ and $\mathsf{TEGC}$ significantly outperform the other algorithms on the running time, and they are the only implemented algorithms running (in minutes) within the time limit for the billion-scale network Twitter. This can be explained by the reason that CELF uses the time consuming monte-carlo sampling method, while the STAB algorithms take a long time for building the sketches. 

In Figs.~\ref{fig:wikiVote_generalCost_Influence_alpha02}-\ref{fig:twitter_generalCost_Influence_alpha02}, we plot the normalized influence spread (i.e., IS/$(\eta-\alpha\eta)$) of the implemented algorithms, where the IS of any solution is evaluated by $10^4$ monte-carlo simulations. It can be seen that CELF, $\mathsf{BCGC}$, $\mathsf{TEGC}$ and STAB-C2 can output feasible solutions with the influence spread larger than $(1-\alpha)\eta$, but STAB-C1 may output infeasible solutions (especially when $\eta$ is large). This phenomenon has also been reported in~\cite{Kuhnle2017}, which can be explained by the reason that the influence spread estimator used in STAB-C1 is less accurate than that in STAB-C2, so it may output solutions with poorer qualities.

From Figs.~\ref{fig:wikiVote_generalCost_Cost_alpha02}-\ref{fig:twitter_generalCost_Cost_alpha02}, it can be seen that the total costs of the solutions output by our algorithms are lower than those of STAB-C2 and CELF, while STAB-C1 has the best performance on the total cost. This is because that STAB-C1 outputs infeasible solutions with the influence spread less than $(1-\alpha)\eta$, so it selects less nodes than the other algorithms. {On the contrary, STAB-C2 selects many more nodes than what is necessary to achieve the $(1-\alpha)\eta$ threshold on the influence spread, and hence it outputs node sets with larger costs than those of our algorithms.}


In Fig.~\ref{fig:the UC setting_alpha02}, we study the performance of the algorithms under the UC setting, where the results are similar to those in Fig.~\ref{fig:the GC setting_alpha02}. In summary, our algorithms (especially ATEUC) greatly outperform the other baselines on the running time, while they also output feasible solutions with small costs. This can be explained by similar reasons with those for  Fig.~\ref{fig:the GC setting_alpha02}.


\section{Related Work}
\label{sec:rw}

Since~\cite{Kempe2003}, a lot of studies have aimed to design efficient $k$SS algorithms. The earlier studies in this line are mostly based on the naive monte-carlo sampling method (e.g.,\cite{LeskovecKGFVG2007,ChenWW2010}), and more recent work~\cite{TangXS2014,TangSX2015,NguyenTD2016,NguyenDT2016} has leveraged more advanced sampling methods~\cite{Borgs2014,DagumKLR1995} to reduce the time complexity, such as the RR-set sampling method. Besides the RR-set sampling method, Cohen~\textit{et.al.}~\cite{CohenDPW2014} have proposed a min-hash sketch based method for $k$SS, which is also used in~\cite{Kuhnle2017}. Some variations of the $k$SS problem have also been studied in \cite{Nguyen2017outward,NguyenDT2016,LinCL17}. However, all these studies belong to the category of influence maximization (IM) algorithms.

Compared with the IM problem, the MCSS problem is less studied in the literature. Recall that we have provided a BA algorithm for MCSS under the GC+NV setting and an AP algorithm for MCSS under the UC+NV setting. To the best of our knowledge, only the work in~\cite{Goyal2013,Kuhnle2017,Chen2014} has provided algorithms with provable performance bounds under the same settings with ours. 
Therefore, we have compared our algorithms with~\cite{Goyal2013,Kuhnle2017,Chen2014} in detail in Sec.~\ref{sec:compareundergc} and Sec.~\ref{sec:compareunderuc}. Some variations of the MCSS problem have also been studied in \cite{DinhZNT2014,ZhangNZT2016,ZhuLZ2016,HeJBC2014}, but the models and problem definitions of these proposals are very different from ours, and none of them has considered the MCSS problem under our setting.


%


\section{Conclusion}
\label{sec:conclu}

We have proposed several algorithms for the Min-Cost Seed Selection (MCSS) problem in OSNs, and compared our algorithms with the state-of-the-art algorithms. The theoretical comparisons reveal that, our algorithms are the first to achieve provable performance bounds and polynomial running time under the case where the nodes have heterogeneous costs, and our algorithms' time complexity outperforms the existing ones in orders of magnitude under the uniform cost case. The experimental comparisons reveal that, our algorithms scale well to big networks, and also significantly outperform the existing algorithms both on the cost and on the running time.

\bibliographystyle{IEEEtran}
\bibliography{IEEEabrv,Mylib}

\appendix

In this section, we provide the missing proofs for Lemma~\ref{lma:thearoftminus1} and Lemma~\ref{lma:atminus1issmall}. The proofs of Lemma~\ref{lma:thearoftminus1} and Lemma~\ref{lma:atminus1issmall} have used Lemma~\ref{lma:alwaysabigone}, which has been proposed in~\cite{Goyal2013}. 

\begin{lemma} [\cite{Goyal2013}]
Suppose that $g(\cdot)$ is a monotone,non-negative submodular function defined on the ground set $V$. Define $H(\Gamma)=\arg\min_{A\subseteq V \wedge g(A)\geq \Gamma}C(A)$. For any $A\subseteq V$ and any $0<\Gamma\leq g(V)$, if $g(A)<\Gamma$, then there must exist $y\in V\backslash A$ such that ${[g(A\cup \{y\})-g(A)]}/{C(\{y\})}\geq {[\Gamma-g(A)]}/{C(H(\Gamma))}.$
\label{lma:alwaysabigone}
\end{lemma}


\begin{proof} [Proof of \textbf{Lemma}~\ref{lma:thearoftminus1}]

Let $B_0=\emptyset$ and $\eta_i=(1-\tau)\eta-\bar{f}(\mathcal{R},B_i)(\forall i\geq 0)$. According to Lemma~\ref{lma:alwaysabigone}, for any $0\leq i\leq s-1$, there must exist $y\in V\backslash B_i$ such that $\bar{f}(\mathcal{R}, B_{i}\cup \{y\})-\bar{f}(\mathcal{R}, B_i)\geq \frac{(1-\tau)\eta-\bar{f}(\mathcal{R}, B_i)}{|D^*|}$. As $\bar{f}(\mathcal{R}, B_{i+1})\geq \bar{f}(\mathcal{R}, B_{i}\cup \{y\})$, we get $\eta_i-\eta_{i+1}\geq \eta_i/|D^*|$ and hence
\begin{eqnarray}
\forall 0\leq i\leq s-1: \eta_{i+1} \leq \left(1-{1}/{|D^*|}\right) \eta_i \label{eqn:ditui}
\end{eqnarray}

Let $\lambda=(1+\frac{1}{n})(1-\tau)\eta - 1$ and $l=\min\{i|0\leq i\leq s\wedge \bar{f}(\mathcal{R},B_i)\geq \lambda\}$ (note that $\lambda<(1-\tau)\eta$). We will prove $l\geq s-1$. Indeed, if $l<s$, then we must have $\bar{f}(\mathcal{R},B_l)<(1-\tau)\eta$ according to the definition of $s$. Moreover, according to Lemma~\ref{lma:alwaysabigone}, there must exist $y\in V\backslash B_l$ such that
\begin{eqnarray}
&&\bar{f}(\mathcal{R},B_l\cup \{y\})-\bar{f}(\mathcal{R},B_l)\geq \frac{n - (1-\tau)\eta}{n}\nonumber\\
&=& (1-\tau)\eta-\lambda;\qquad \label{eqn:zengliang}
\end{eqnarray}
which implies $\bar{f}(\mathcal{R},B_l\cup \{y\})\geq (1-\tau)\eta$ and hence $s\leq l+1$.

If $l=0$, then we have $s\leq 1$ and hence the lemma trivially holds. In the sequel, we assume $l\geq 1$. Note that $\bar{f}(\mathcal{R},B_{l-1})<\lambda$ according to the definition of $l$. Using equations (\ref{eqn:ditui}) and (\ref{eqn:zengliang}), we can get
\begin{eqnarray}
&&(1-\tau)\eta -\lambda <\eta_{l-1}\leq \left(1-{1}/{|D^*|}\right)^{l-1} (1-\tau)\eta \nonumber\\
&\leq& \exp\left\{-\frac{l-1}{|D^*|}\right\}(1-\tau)\eta
\end{eqnarray}
and hence $l-1<|D^*|\ln\frac{(1-\tau)n\eta}{n-(1-\tau)\eta}\leq |D^*|\ln\frac{n\eta}{n-\eta}$.
As $l$ is an integer, we get
\begin{eqnarray}
l \leq \left\lceil |D^*|\ln\frac{n\eta}{n-\eta}\right\rceil\leq |D^*|\left\lceil \ln\frac{n\eta}{n-\eta}\right\rceil \label{eqn:thelaststep}
\end{eqnarray}
The lemma then follows by combining (\ref{eqn:thelaststep}) with $|B_s|=s\leq l+1$.
\end{proof}

\begin{proof} [Proof of \textbf{Lemma}~\ref{lma:atminus1issmall}]
If $\bar{f}(\mathcal{R},B_{s})\geq (1+\tau)\eta$, then the lemma trivially holds. Now suppose $\bar{f}(\mathcal{R},B_{s})< (1+\tau)\eta$. As $\tau\leq\frac{n-\eta}{2n\eta+\eta}$, we must have $\bar{f}(\mathcal{R},B_{s})< (1+\tau)\eta\leq n$. Let $A^*\subseteq V$ be the set with the minimum cost that satisfies $\bar{f}(\mathcal{R},A)\geq n$. 
According to Lemma~\ref{lma:alwaysabigone}, there must exist $y\in V\backslash B_{s}$ such that
\begin{eqnarray}
&&\bar{f}(\mathcal{R},B_{s}\cup \{y\})-\bar{f}(\mathcal{R},B_{s})\geq \frac{n - \bar{f}(\mathcal{R},B_{s})}{C(A^*)}\nonumber\\
&\geq& \frac{n - (1+\tau)\eta}{n}\geq 2\tau\eta, \nonumber 
\end{eqnarray}
where the last two inequalities are due to $C(A^*)\leq n$ and $\tau\in (0,\frac{n-\eta}{2n\eta+\eta}]$. This implies $\bar{f}(\mathcal{R},B_{s+1})\geq \bar{f}(\mathcal{R},B_{s}\cup \{y\})\geq (1+\tau)\eta$. Therefore, we have $t\leq s+1$. Hence the lemma follows.
\end{proof}

\end{document}